\documentclass[12pt]{article}
\usepackage[top=1.25in, bottom=1.25in, right=1.25in,left=1.25in]{geometry}

\usepackage{amsthm,amsmath,amssymb,amsfonts,mathrsfs,amssymb}
\usepackage{natbib}
\usepackage{upgreek}
\usepackage{bm}
\usepackage{color}
\usepackage{tikz,epstopdf}
\usetikzlibrary{arrows,decorations.pathreplacing, patterns}
\usepackage[titletoc, title]{appendix}
\usepackage{enumerate}
  \usepackage{graphicx}
\newcommand\smallO{
  \mathchoice
    {{\scriptstyle\mathcal{O}}}
    {{\scriptstyle\mathcal{O}}}
    {{\scriptscriptstyle\mathcal{O}}}
    {\scalebox{.7}{$\scriptscriptstyle\mathcal{O}$}}
  }
\usepackage{subfig}
   \usepackage{array}
   \usepackage{cases}
   \usepackage{setspace}
   
\usepackage{verbatim}
\usepackage[nocomma]{optidef}

\newcommand{\C}{\mathcal{C}}

\newtheorem{thm}{Theorem}
\newtheorem{lem}{Lemma}
\newtheorem{prop}{Proposition}
\newtheorem{col}{Corollary}
\theoremstyle{definition}
\newtheorem{defn}{Definition}
\newtheorem{asmp}{Assumption}
\newtheorem{exam}{Example}

\usepackage{hyperref}
\hypersetup{colorlinks=true,
            citecolor=blue,
            linkcolor=blue,
            urlcolor=blue,
			bookmarksopen=true,
			pdfstartview={XYZ null null 1.00},
			pdfpagelayout={SinglePage}
}
\usepackage{url}

\begin{document}

\title{Incentivizing Information Acquisition%
	\thanks{%
		\scriptsize{I am indebted to my committee chair Luciano Pomatto and advisor Omer Tamuz for their generous encouragement, guidance, and support. This paper also benefits from insightful discussions with Alexander Bloedel, Peter Caradonna, Axel Niemeyer, and Kun Zhang. For helpful comments, I also thank S. Nageeb Ali, Pak Hung Au, Daghan Carlos Akkar,  Simon Board,   Simone Cerreia-Vioglio, Daniel Gottlieb, Simon Grant, Nima Haghpanah, Wei He, Jonathan Libgober, Shuo Liu, Yi Liu,  Jay Lu, Yusufcan Masatlioglu, Doron Ravid, Fedor Sandomirskiy, Balazs Szentes, Xin Shan, Andrew J. Sinclair, Jakub Steiner, Rakesh Vohra,  Mu Zhang, Mengxi Zhang, Weijie Zhong, seminar participants at Caltech, UCLA, HKU, HKUST, PKU, THU,  and conference participants at 2025 CETC, 2024 AMES, 2024 Midwest Theory Meeting.}
}

\author{Fan Wu\thanks{
		\scriptsize{Fan Wu: HSBC Business School, Peking University (email: fanwu@phbs.pku.edu.cn).}}
  }
}

\maketitle
\begin{abstract}
I study a principal-agent model in which a principal hires an agent to collect information about an unknown continuous state. The agent acquires a signal whose distribution is centered around the state, controlling the signal's precision at a cost. The principal observes neither the chosen precision nor the signal. Instead, she commits to a transfer rule, subject to limited liability and a hard budget, that can depend on the agent’s report and the state realized ex post. A simple transfer rule has a cutoff structure: the agent receives a fixed prize if his report is sufficiently close to the state and receives nothing otherwise. I characterize exactly when cutoff transfers are optimal for all cost functions. The condition is weaker than MLRP and is satisfied by standard signal distributions commonly used in the literature.

\medskip


\bigskip
\end{abstract}
\newpage

\section{Introduction}

In many situations, an agent is tasked with collecting information about a quantity of interest. Examples include a statistician collecting data for the Bureau of Statistics, a meteorologist forecasting weather for the National Weather Service, or a consultant assessing the profitability of a new market. A natural question is how to design a contract to incentivize the agent to collect information and report it accurately.

 A difficulty in designing such contracts is that the agent's effort and the information acquired are often unobservable to the principal. For example, a statistician might be able to fabricate part of their data. Similarly, in the case of a meteorologist analyzing weather patterns, the amount of time and effort spent on the analysis might not be easily observable. In the case of the consultant, predicting market profitability may rely on private connections or specialized skills, which are often proprietary. Therefore, the contract between the principal and the agent must take into account moral hazard.

In this paper, I study the design of an optimal contract for a principal  (she) that incentivizes the agent (he) to gather information and maximize its precision. The agent conducts a location-scale experiment that is centered around the state and controls the precision (the inverse of the scale) at a cost. Eventually, the state is revealed and the contract can depend on the state and the agent's report. I abstract from the principal's broader decision problem---that is, the way in which the principal uses the information generated by the agent---and focus on the problem where the principal's objective is only to maximize precision subject to some fixed budget constraint on transfers. My analysis and results shall be robust to whatever the broader decision problem may be. The budget reflects the maximum amount of transfer that the principal is able to use.

In practice, a budget that can only be used for a specific task is commonly observed. For example, government agencies typically operate within budget limits imposed by government funding and must return any unused funds to the government.\footnote{\cite{schick2008federal} provides an in-depth look at how the federal budget is structured and constrained by funding limits imposed by the government.} There are numerous reports of unspent government funding being returned to the government, such as \cite{Millionsremainunspent}, \cite{DCleftmillions}.\footnote{ Federal regulations also specify where unspent grant money is returned, as seen at the \cite{codeoffederal}.} In firms, it is also common for agents to operate under a budget and to eventually return unused funds to the firm. \citet{anthony2007management} and \citet{horngren2009cost} provide comprehensive empirical evidence for this practice.\footnote{\citet{anthony2007management} offer an extensive analysis of organizational budget management, highlighting the frequent practice of returning unspent funds to a central or general fund in both private and public sectors. \citet{horngren2009cost} explores cost control and budget allocation, noting how unused resources are often returned to a central pool for future use.} In the consulting example, the transfer could take the form of a promotion, a recommendation, or a grade. Sometimes, budgets are restricted by external funding designated for specific purposes, such as university grants.

 In such settings, a simple contract is to allocate the entire budget to the agent if their report is close to the actual state, and withhold payment otherwise. I call these incentive schemes \textit{cutoff transfers}. The cutoff transfer is the simplest transfer rule: It is straightforward to understand and easy to implement. The literature on principal-agent problems has long been interested in when simple contracts are optimal \citep{carroll2015robustness,Herweg2010,oyer2000quota,gottlieb2022simple}. 

 My main result is that cutoff transfers are optimal across a large variety of settings. In particular, I identify a sufficient and necessary condition on the agent's information structure such that for every cost function, there exists an optimal transfer that is a cutoff transfer.

 Formally, I study a model in which the principal wants to incentivize the agent  to acquire information on an underlying state $\theta \in \mathbb{R}$. For simplicity, I assume the state $\theta$ admits an improper uniform common prior.\footnote{I generalize the model to the case where the state is multi-dimensional and the case with Gaussian prior.} The agent can acquire a costly signal $s\in \mathbb R$ about the state. The signal takes the form $s=\theta+\frac{1}{\lambda}\varepsilon$ where $\varepsilon$ is drawn from some symmetric and single-peaked probability density function, and $\lambda$ is a measure of precision controlled by the agent. An important example is the case where $\varepsilon$ admits a standard Gaussian distribution. The agent chooses precision $\lambda$ at a cost $c(\lambda)$. 

The principal observes neither the agent's signal $s$ nor the agent's choice of precision $\lambda$. The agent makes a report $a\in \mathbb R$ after observing the signal. The principal wants to incentivize the agent to maximize precision and truthfully report his signal. She does so by a transfer rule $t$ that depends on the report $a$ and the state $\theta$, which is eventually observed by both players.\footnote{I study the case where the state is unobservable but the principal later receives a private signal about the state.}  I impose limited liability ($t\geq 0$) and limited budget. The principal has no use for the unspent savings. The agent's payoff is his expected utility from transfers minus the cost. I allow for any increasing utility function, which accommodates arbitrary risk attitude of the agent. I also allow for arbitrary lower-semicontinuous cost function.

I refer to a transfer rule as \textit{optimal} if it induces the maximum precision among all transfers and elicits truthful reports. There are at least two difficulties in this problem. First, the transfer design problem is an infinite-dimensional optimization problem, as the transfer rule itself is a function. Second, there are two layers of incentive compatibility issues. The principal does not observe the precision chosen by the agent, which is a moral hazard problem, and she does not observe the signal either, leading to a elicitation problem.

I say that a transfer rule is a \textit{cutoff} transfer if it pays the entire budget when the distance between the report and the state is below a cutoff $d$ and pays $0$ otherwise. For which signal distributions are cutoff transfers optimal?

My main result (Theorem \ref{thm: step}) identifies a sufficient and necessary condition for cutoff transfers to be optimal. This condition is weaker than the monotone likelihood ratio property \citep{milgrom1981good,rogerson1985first,jewitt1988justifying} and is satisfied by most common distributions, including Gaussian, Laplace, logistic, and the uniform distribution. Theorem \ref{thm: step} shows that if this condition holds, then for all cost functions cutoff transfers are optimal. Conversely, if the condition does not hold, then there is a cost function for which no cutoff transfer is optimal. Moreover, this cost function is not pathological and can be taken to be increasing, convex, and differentiable. Furthermore, when the condition holds, I characterize the optimal cutoff (Theorem \ref{thm: optimal step function}).

To see why only part of the error distribution enters the condition, consider raising the reward at a small forecast error. If that error becomes more likely as precision rises, the extra reward favors the target precision over lower precision choices. One can therefore pay the full prize throughout the range of small errors with this property. Among the remaining errors, increasing elasticity orders outcomes by how strongly they favor higher precision. Moving rewards toward the more accurate forecasts then yields an interval. Notably, the proofs rely on techniques from monotone comparative statics. This allows me to analyze that problem under weak assumptions on the cost function and signal distribution. In particular, I do not need to assume the validity of the first-order approach \citep{rogerson1985first,jewitt1988justifying}.

All the results generalize to the $n$-dimensional case. The signal distribution is still symmetric, single-peaked, and centered around $\theta$. Here the symmetry means that the density of the distribution depends only on the Euclidean distance from the state $\theta$. A cutoff transfer pays $1$ when the Euclidean distance between the report and the state is less than a cutoff $d$. In this setting I show that an analogous result holds (Proposition \ref{prop: high dimensional}). 

I also generalize my results to the case of a proper Gaussian prior and Gaussian signals. I show that for all cost functions, cutoff transfers are optimal (Proposition \ref{prop: Gaussian}). 

In addition, I extend my results to the setting where the state is unobservable. Instead, the principal privately receives a signal about the state. When designing the transfer, she uses her signal instead of the state to discipline the agent. I show that given a Gaussian or uniform prior and Gaussian signal, cutoff transfers are optimal (Proposition \ref{prop: unobserved}).

Lastly, I examine what happens when the principal values savings from the unused budget. In Section~\ref{sec: savings}, she maximizes the benefit of precision minus the expected transfer, subject to the same hard budget constraint, with a risk neutral agent. The problem can be decomposed into minimizing the expected transfer for each target precision and then choosing the target. A transfer rule is \textit{optimal} for a given precision target if it minimizes the expected transfer among all transfer rules and induces truthful report. Theorem~\ref{thm: cost minimizing cutoff} identifies a sufficient and necessary condition for cutoff transfers to be optimal. This condition is different from the one from precision-maximization and is also weaker than the monotone likelihood property. When this condition holds, for every continuously differentiable, increasing, and convex cost function and every implementable precision target, cutoff transfers are optimal. Conversely, if this condition does not hold, then there is a linear cost function for which no cutoff transfer is optimal. Notably, the intersection of the condition for precision maximization (Theorem~\ref{thm: step}) and the condition for expected transfer minimization (Theorem~\ref{thm: cost minimizing cutoff}) is exactly MLRP. I also characterize the implementable targets, optimal cutoff, and the minimum expected transfer (Theorem~\ref{thm: minimum expected transfer}).

\subsection{Related Literature}
My paper is related to the literature on incentivizing information acquisition \citep[see, e.g.,][]{osband1989optimal,Liandlibgober2023,li2022optimization,neyman2021binary,zermeno2011,carroll2019robust,chen2023learning,chade2016delegated,whitmeyer2023buying,sharma2023procuring,clark2021contracts}.\footnote{In \cite{neyman2021binary}, a rational expert aims to predict the probability of a biased coin flip. The information is acquired by choosing the number of flip trials at a fixed cost per flip. \cite{neyman2021binary} study the optimal scoring rule that incentivizes precision. }

In \cite{li2022optimization}, an agent exerts a binary level of effort to refine a posterior from a prior. The paper studies optimizing proper scoring rules by maximizing the increase in the score with effort. In \cite{Liandlibgober2023}, a principal hires an agent to learn about a binary state. The agent acquires information over time through a Poisson information arrival technology. The principal rewards the agent with a fixed-value prize as a function of the agent's sequence of reports and the state. \cite{Liandlibgober2023} identify conditions under which it is without loss to elicit a single report after all the information has been acquired.

\cite{osband1989optimal} studies a principal with a quadratic prediction cost who incentivizes an expert to collect information. The expert can increase the precision (the inverse of variance) of the prediction at a constant cost. The principal minimizes the sum of the expected error variance and the expected transfer. In this stylized setting, the optimal transfer consists of a quadratic report error term plus a linear term on an initial belief error, with three parameters.

In \cite{whitmeyer2023buying} and \cite{sharma2023procuring}, a rational inattentive agent can acquire information flexibly subject to a posterior separable cost. The principal wants to minimize the expected monetary cost of implementing a given information structure. Similarly, \cite{clark2021contracts} study the Pareto optimal contract that maximizes social welfare. In all these papers, both the experiment and the signal are unobservable to the principal. In \cite{rappoport2017incentivizing}, a principal hires an agent to acquire costly information to influence a third party's decision. This paper assumes that the realized piece of information is observable and contractible.

\cite{zermeno2011} and \cite{papireddygari2022contracts} study menu design with information acquisition. In their setting, the principal first offers a menu of contracts. Then the agent privately acquires costly information. Next, the agent selects a contract from the menu. The selection therefore reveals some information the agent acquired.

In \cite{carroll2019robust}, the principal is uncertain about the expert's information acquisition technology and only knows some experiments that the agent can choose. The principal evaluates the incentive contract by a worst-case criterion.

\cite{argenziano2016strategic} and \cite{kreutzkamp2022endogenous} study costly information acquisition and transmission. In both papers, there is no transfer and the agent cares about the principal's action. In \cite{kreutzkamp2022endogenous} setting, the sender publicly chooses an experiment. In \cite{argenziano2016strategic}, the expert acquires information by choosing a number of binary trials to perform.

In the classic principal-agent model, several papers also show the optimality of the cutoff transfer but rely on different assumptions. \cite{oyer2000quota} assumes the validity of the first-order approach, the existence of the optimal contract, and the absence of the IR constraint. He shows that cutoff transfers are optimal among monotone contracts for a risk-neutral agent. \cite{Herweg2010} show that cutoff transfers are optimal for expectation-based loss-averse and risk-neutral agent.

The extension in Section~\ref{sec: savings} also connects this paper
to the classic principal-agent literature. \citet{Grossman1983}
decompose the principal's problem into two stages: minimizing the
expected cost of implementing each action, subject to global
incentive compatibility and participation, and then choosing the
action that maximizes expected benefits net of implementation costs.
In their finite-outcome model, the implementation problem can be
formulated as a convex program. Their analysis shows that MLRP
alone need not imply a monotone optimal transfer.
\citet{jewitt2008moral} study implementation costs and optimal
contracts under upper and lower bounds on payments.

\section{The Model}
I study a principal-agent model, where the principal (she) wants to incentivize the agent (he) to acquire information regarding an underlying state $\theta\in \mathbb R$. They share a common prior over $\theta$. It will be convenient to assume the prior to be an improper uniform prior on  $\mathbb R$. I generalize my results to the case in which the state is multi-dimensional in Section \ref{sec: multidimensional}, and I study the case of a  
proper Gaussian prior in Section \ref{sec: Gaussian prior}.

 The agent can acquire a costly signal $s\in \mathbb R$ of the form
 $$
    s =\theta+\frac{1}{\lambda}\varepsilon,
 $$
 where $\varepsilon$ is drawn from a distribution with a symmetric and single-peaked probability density function (PDF) $\phi$. I assume $\phi$ is continuously differentiable. The function $\phi$ is supported on an interval, which can be bounded or unbounded. The parameter $\lambda$ is a measure of the signal's precision, a scale parameter that is inversely proportional to the standard deviation of $s$. The PDF of the signal, given $\theta$ and $\lambda$, is denoted by $\varphi(\cdot; \theta, \lambda)$. It takes the form
 $$
    \varphi(x;\theta, \lambda)=\lambda \phi(\lambda (x-\theta)).
 $$
 For example, if $\phi$ is the PDF of the standard Gaussian distribution, then $\varphi(\cdot;\theta,\lambda)$ is the PDF of a Gaussian distribution with mean $\theta$ and standard deviation $1/\lambda$. Note that $\phi(x) = \varphi(x;0,1)$.

The agent chooses the precision $\lambda$ at a cost $c(\lambda)$. For example, suppose the agent's signal is the aggregation of many small independent signals. Then as the number of small signals becomes large, the aggregate signal tends to a Gaussian distribution. If the agent incurs a cost that depends on the number of small signals he gathers, then the agent effectively controls the standard deviation of the aggregate signal at a cost. I assume that the cost function $c$ is lower semicontinuous.\footnote{It would be natural to also assume that the cost function is increasing---that is, the higher the precision, the higher the cost. Imposing this assumption or not does not affect my results.} The pair $(\phi,c)$ constitutes the primitive of the model. 

The principal observes neither the agent's signal $s$ nor the agent's choice of precision $\lambda$. Instead, after observing the signal, the agent sends a report $a\in \mathbb R$ to the principal.\footnote{It is without loss for the agent to report his signal only. Asking the agent to report the precision additionally would only introduce cheap talk. The agent shall report the precision that results in the largest expected transfer.} Eventually both players observe the state $\theta$.\footnote{I relax this assumption in Section \ref{sec: unobserved state} where instead of eventually observing the state, they eventually observe a signal about the state.} The principal wants to incentivize the agent to maximize precision and truthfully report his signal. She does so by means of a transfer.  The transfer can depend on the state and the agent's report. The assumption that the realized state is contractible is familiar from the literature on belief elicitation via (proper) scoring rules and prediction markets, ubiquitous in the principal-expert literature, and well-suited to economic applications in which the state is publicly observable ex post (e.g., the temperature tomorrow, the electricity demand next week, the outcome of an election being forecast by a pollster, or the conditions of a new market being analyzed by a consultant).

I assume that the transfer $t$ is a function of the difference $\theta-a$ and it vanishes at infinity. That is, it satisfies
\[
    \lim_{x \to -\infty} t(x) = \lim_{x \to \infty} t(x) = 0.
\]
As I discuss below, the assumption that the transfer  depends only on the difference between $a - \theta$ is to ensure the agent's payoff is well-defined. I also assume limited liability, i.e., that the principal cannot take money from the agent, so that $t\ge 0$. In addition, the principal has a limited budget that does not depend on the state. She has no use for the unspent budget.\footnote{Notably, there is a growing literature that adopts the same assumption that the principal cannot take money from the agent and can only reward the agent with a prize for which they have no other uses; see \cite{Liandlibgober2023,li2022optimization,deb2018evaluating,deb2023indirect,dasgupta2023optimal,hebert2022engagement,wong2023dynamic}.} I consider the case where the principal values savings in Section~\ref{sec: savings}.  Without loss, I take the budget to be $1$.

The timing of moves is as follows:
\begin{enumerate}
    \item The principal commits to a transfer rule $t\colon \mathbb R\rightarrow \mathbb [0,1]$.
    \item The agent chooses to accept or not. The game continues if the agent accepts and terminates otherwise.
    \item The agent privately chooses a signal precision $\lambda$.
    \item The agent privately observes a signal realization $s$ drawn from the distribution $\varphi(\cdot; \theta, \lambda)$.
    \item The agent makes a report $a\in \mathbb R$.
    \item Finally, the state $\theta$ is revealed and the agent receives the transfer $t(\theta-a)$.
\end{enumerate}

 The agent's payoff is given by the expected utility of the transfer minus the cost of acquiring information:
 \[
    \mathbb E u(t)-c(\lambda),
 \]
 where $u \colon [0,1] \to \mathbb{R}$ is strictly increasing and continuous, and normalized so that $u(0)=0$, and $u(1)=1$.\footnote{Here I can drop the continuity assumption on $u$ and only require $u$ to be weakly increasing. All my results still hold.} To ease the notation, it will be convenient to treat the transfer rule as paying in utils, rather than money. Under my assumptions on $u$, this is equivalent, and without loss of generality, to assuming $u$ is the identity function. Finally, I assume the agent has an outside option  which gives payoff $0$ if he rejects the contract.

 After choosing $\lambda$, the agent observes a signal. Conditional on a signal realization $s$, it follows from Bayes' rule (adapted to a uniform prior) that the agent's posterior has PDF $\varphi(\cdot; s,\lambda)$. The agent now chooses a report $a$ to maximize the expected transfer, which is therefore given by\footnote{I set $E(0;t)=0$, which agrees with the limit as precision tends to zero.}
$$
    E(\lambda;t) = \max_{a \in \mathbb{R}} \int_{\mathbb R} t(\theta'- a) \varphi(\theta'; s,\lambda) d\theta'.
$$
Maximizing over report $a$ shows that when choosing the precision, the agent anticipates how he optimally reports later. Thus, the model allows for double deviation. The assumption that $t$ vanishes at infinity ensures that the maximum is well defined.\footnote{I prove this in Lemma \ref{lem: expected transfer} in the appendix.}

Note that as $t$ depends only on the difference $\theta-a$, the expected transfer computed above does not depend on $s$.\footnote{To see this, note that \[
\begin{split}
    &\max_a \int_{\mathbb R} t(\theta'- a) \varphi(\theta'; s,\lambda) d\theta'
    =\max_a \int_{\mathbb R} t(\theta'- a+s) \varphi(\theta'; 0,\lambda) d\theta'
    =\max_a \int_{\mathbb R} t(\theta'- a) \varphi(\theta'; 0,\lambda) d\theta'
\end{split}
\]} Consequently, the agent's (interim) expected transfer conditional on signal $s$ does not depend on $s$, which must coincide with his unconditional expected transfer. If $t$ depends arbitrarily on $a$ and $\theta$, the agent's unconditional expected transfer might not be well-defined, since the distribution of $s$ is improper. In Appendix \ref{sec: proper uniform},  I show that when the agent's expected transfer is well-defined, it is without loss to assume that the transfer depends only on the difference $a-\theta$.  An interesting fact is that almost all popular key performance indicators in the forecast industry depend only on the prediction error $a-\theta$ \citep{vandeput2021data}.

The uniform prior assumption is mainly a translation-invariance
device: together with the location-scale signal structure, it makes the posterior distribution of the prediction error independent of the realized signal, up to
the precision parameter. This invariance allows for a sharp characterization of the optimal contract. With an
arbitrary prior, the posterior error distribution would generally depend on the
realized signal, so the optimal contract would typically be
state-dependent, which is beyond the scope of this paper. I return to proper Gaussian priors in Section~5.2, where conjugacy restores an
effective location-scale representation.

The principal's objective is to maximize $\lambda$ under the constraint of inducing a truthful report from the agent. (As I show later, this constraint will not be binding.) This objective can be interpreted as capturing, in reduced-form, settings in which the principal uses the information conveyed by the agent's report to solve some (un-modeled) decision problem. Indeed, under broad conditions, it is optimal for the principal to maximize precision regardless of the decision problem that she faces. First, when the density $\phi$ is strongly unimodal, for a large class of monotone decision problems \citep{karlin1956theory}, higher precision is always better for the principal \citep[see][Theorem 5.1 and 5.2]
{lehmann2011comparing}.\footnote{$\phi$ is strongly unimodal if $-\ln\phi$ is convex. This is slightly stronger than global increasing elasticity.}\footnote{The class of \textit{monotone} decision problems is defined in terms of the action space and the permissible loss functions. For each $\theta$, there is a correct action $A(\theta)$. The function $A(\theta)$ is real-valued and nondecreasing. The range of $A(\theta)$ is the action space. The loss function is minimized at the correct action and is nondecreasing as the action moves away from the correct action on either side.}  Second, when the distribution $\phi$ is \textit{self-decomposable}, the signal with different precisions are ranked in the Blackwell order; thus the principal is better off with higher precision for all decision problems. The self-decomposable distributions include all stable distributions (such as Gaussian and Cauchy distributions) and some non-stable ones such as Laplace.\footnote{See \cite{goel1979comparison} and \cite{lehmann2011comparing} for a detailed discussion.}

\section{Preliminary Analysis}\label{sec: prelim}

First, I define a relaxed problem, in which the principal designs a transfer to maximize the agent's signal precision, without incentivizing truthful reporting.
\begin{maxi}
{\lambda, t}   
\lambda  
{\label{eq:program}}  
{}  
\addConstraint{0}{\le t\le 1}
\addConstraint{\lambda}{\in \arg\max E(\cdot;t)-c(\cdot)}{\quad\text{IC}}
\addConstraint{E}{(\lambda;t) -c(\lambda)\ge 0}{\quad\text{IR}}
\end{maxi}
The first constraint in the relaxed problem is limited liability and limited budget. The second constraint is an incentive compatibility constraint, since the precision $\lambda$ is chosen by the agent and is unobserved by the principal and cannot be contracted upon.\footnote{The continuity of $E(\cdot; t)$ and the lower semicontinuity of $c$ ensures that the $\arg\max$ in IC is well defined under very weak assumptions. For instance, the set of $\lambda$ is compact or $\lim_{\lambda\to +\infty}c(\lambda)>1$.} The last constraint is a participation constraint, also known as individual rationality, and is implied by the outside option available to the agent. A transfer rule $t$ is \textit{optimal} if it solves problem \eqref{eq:program} and induces the agent to report truthfully.

A quantity that will be key to the analysis is the elasticity of the standardized signal distribution $\phi$.
\begin{defn}
    The \emph{elasticity} $\eta$ of signal distribution $\phi$ at $x > 0$ is defined as
    \[
        \eta (x)=-\frac{d  \phi(x)/\phi(x)}{d x/x}.
    \]
    For $x$ such that $\phi(x)=0$, I let $\eta (x)= +\infty$.
\end{defn}
Elasticity measures by how much a percentage change in $x$ leads to a percentage change in the density function $\phi(x)$. For example, if $\phi$ is the PDF of a standard Gaussian, then $\eta(x)=x^2$. If $\phi$ is the PDF of a standard Laplace, then $\eta(x)=x$. Notice that both examples have weakly increasing elasticity. In fact, weakly increasing elasticity is satisfied by most common nonatomic distributions defined on an interval, including the uniform distribution, the triangular distribution, and the logistic distribution. The next lemma shows that the weakly increasing elasticity condition is equivalent to a monotone likelihood property: for any $0\leq x_1\leq x_2$, the ratio 
$\frac{\varphi(x_1;0,\lambda)}{\varphi(x_2;0,\lambda)}$ increases in $\lambda$. That is, a more accurate signal ($x_1$) is more likely to appear than a less accurate signal ($x_2$) as precision increases.
\begin{lem}\label{lem: ratio}
$$\frac{\partial [\ln \phi(\lambda x)]}{\partial\ln \lambda}=-\eta (\lambda x).$$
In particular, for any $0\leq x_1\leq x_2$,
$$\frac{\varphi(x_1;0,\lambda)}{\varphi(x_2;0,\lambda)}\text{ increases in }\lambda \quad \Leftrightarrow \quad \eta (\lambda x_1)\leq \eta (\lambda x_2).$$ 
\end{lem}

Global increasing elasticity is sufficient but not necessary; only the tail beyond the point where elasticity exceeds 1 matters. Formally, I define a condition weaker than MLRP.
\begin{defn}
The signal distribution satisfies \textit{increasing elasticity above 1} if
\[
    x\longmapsto\max\{\eta(x),1\}
\]
is weakly increasing.
\end{defn}
I define a useful quantity.
$$\eta ^{-1}(1)= \inf\{x\in \mathbb R_+|\eta (x) >  1\}.$$
Note that $\eta ^{-1}(1)$ is well-defined, as the density function $\phi$ is integrable.\footnote{Notice that function $1/x$ is not integrable over any neighborhood around $0$. Thus, $\eta (0^+)<1$ since $\phi$ is integrable over any neighborhood around $0$. Moreover, function $1/x$ is not integrable over any neighborhood around $+\infty$. Thus, $\eta (x)>1$ must hold for some $x\in \mathbb R_+$.} Later on, when I generalize the state and the signal to be $n$ dimensional, $\eta ^{-1}(n)$ is defined similarly.

\section{Characterization of the Optimal Transfer}\label{sec: main result}

A simple transfer rule is the \textit{cutoff transfer} that pays the agent $1$ when the distance between the report and the state is less than a cutoff $d$ and pays $0$ otherwise (see Figure \ref{figure step}). As cutoff transfers are symmetric and single-peaked, they have the additional desirable property of inducing truthful reports by the agent.
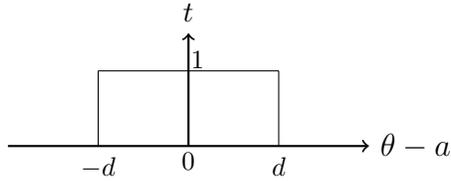
\begin{figure}
    \centering
    \begin{tikzpicture}[xscale=0.6,yscale=0.5]
        \draw[thick, ->] (0,0) -- (0,3) node[above]{\small $t$};
        \draw[thick, ->] (-4,0) -- (4,0) node[right]{ $\theta-a$};
        \draw  (-2,2) -- (2,2);
        \draw  (2,0) -- (2,2);
        \draw  (-2,0) -- (-2,2);
        \node at (0,-0.4) {\footnotesize 0};
        \node at (0.2,2.3) {\footnotesize $1$};
        \node[below] at (2,0) {\footnotesize $d$};
        \node[below] at (-2,0) {\footnotesize $-d$};
    \end{tikzpicture}
    \caption{Cutoff Transfer.} \label{figure step}
\end{figure}

As transfer rules are functions, the principal faces an infinite dimensional optimization problem. In my main result I show that when the standardized distribution $\phi$ satisfies increasing elasticity above one---which holds for all commonly used signal distributions---the principal chooses a cutoff transfer rule, reducing the problem to a one-dimensional one. Moreover, I show that this property of $\phi$ is also necessary for cutoff transfer rules to always be optimal.

\begin{thm}\label{thm: step}
The following are equivalent.
    \begin{enumerate}
        \item For every cost function, there exists an optimal transfer that is a cutoff transfer.
        \item For every increasing, convex, and continuously differentiable cost function, there exists an optimal transfer that is a cutoff transfer.
        \item Signal density function $\phi$ satisfies increasing elasticity above $1$.
    \end{enumerate}
\end{thm}

 Note that Theorem~\ref{thm: step} holds even if we assume that cost functions are differentiable, convex, and increasing; that is, the result is not driven by considering ill-behaved cost functions.

Once we verify the condition in statement (3), then we do not need to worry about the infinite dimensional optimization problem, as cutoff transfers shall be optimal. Moreover, the condition is very easy to check. We only need to compute the elasticity and check if it is monotone in some region. A direct implication of Theorem 1 is that for every cost function, there exists an optimal transfer that is a cutoff transfer when the signal distribution is Gaussian, since the elasticity of a Gaussian distribution is increasing. The lower bound $1$ in the statement is  the dimension of the problem. So far both the state and the report are one-dimensional. In Section \ref{sec: multidimensional}, where I generalize this result to $n$-dimensional signal and states, the dimension $n$ will replace $1$.

 Note that there is a unique solution---which must be a cutoff transfer---if we additionally impose three mild assumptions.\footnote{I prove this uniqueness result at the end of the Proof of Theorem \ref{thm: step}.}
  \begin{enumerate}
        \item Signal density function $\phi$ satisfies strictly increasing elasticity above\footnote{Density function $\phi$ satisfies \textit{strictly increasing elasticity above 1} if $\eta (\cdot)$ single-crosses $1$ from below, $\{x|\eta(x)=1\}$ is a singleton, and is strictly increasing after the cross.} $1$.
        \item The cost function is continuously differentiable.
        \item The optimal precision is an interior solution.
    \end{enumerate}

In the following, I shall first illustrate the sufficiency part (i.e., why (3) implies (1)). Second, I optimize over cutoff transfers and characterize the optimal cutoff. Third, I provide intuition for the necessity part (i.e., why (2) implies (3)).

I define the agent's expected payoff given transfer $t$ as
 $$\pi(t)= \max_\lambda E(\lambda;t)-c(\lambda)$$
 and the agent's choice of precision given $t$
 \[\bm{\lambda}(t)=\begin{cases}
 \max [\arg\max_\lambda E(\lambda;t)-c(\lambda)] &\text{if } \pi(t)\geq 0,\\
 0 &\text{otherwise.}\end{cases}\]
 When analyzing cutoff transfers, with slight abuse of notations, I use $E(\lambda; d)$ to represent the expected transfer at precision $\lambda$ given the cutoff transfer with cutoff $d$, i.e., $E(\lambda; d)= E(\lambda; \textbf{1}_{|\theta-a|\leq d})$. Similarly, I denote by $\pi(d)$ and $\bm{\lambda}(d)$ the expected payoff and choice of precision for the cutoff transfer $d$, respectively.

\subsection{Cutoff Transfers Are Optimal}\label{subsec: sufficiency}

 I explain why increasing elasticity above 1 implies that for all cost functions, cutoff transfers are optimal. Suppose that increasing elasticity above 1 holds. I show that for any transfer rule $t$, there exists a cutoff transfer that induces a weakly larger precision and truthful report. 
 
 Here, for ease of exposition, I assume that $\eta$ is weakly increasing and consider a transfer rule $t$ that is symmetric in $\theta-a$ and weakly decreases in $|\theta-a|$. Given such a transfer rule, the agent reports truthfully and chooses the precision $\bm\lambda(t)$. I construct a cutoff transfer by choosing a cutoff $d$ such that
\begin{equation}\label{eq: equal}
E(\bm\lambda(t);d)=E(\bm\lambda(t);t).
\end{equation}
Such a $d$ exists because $E(\bm{\lambda}(t);d)$ increases in $d$ from $0$ to $1$ (as the budget is 1). See Figure \ref{pic: Delta t} for an example.
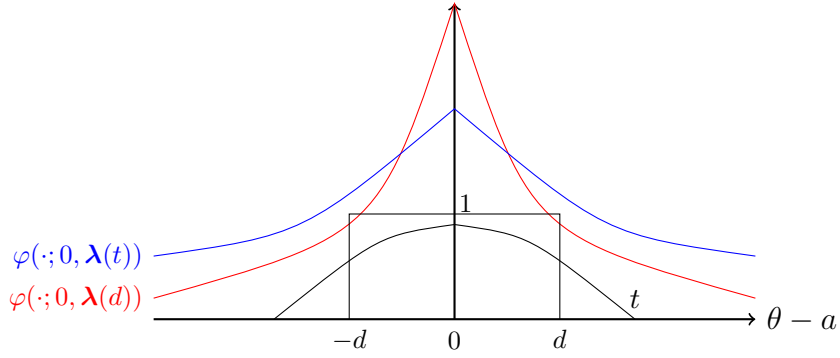
\begin{figure}
    \centering
    \begin{tikzpicture}[xscale=1, yscale=1.4]
         \draw[thick, ->] (0,0) -- (0,3);
        \draw[thick, ->] (-4,0) -- (4,0) node[right]{ $\theta-a$};
        \draw (2.4,0)node[above]{\small $t$}  .. controls (1,0.8) .. (0,0.9)..controls (-1,0.8)..(-2.4,0);
        \draw[red] (4,0.2)   .. controls (1,0.8) .. (0,3)..controls (-1,0.8)..(-4,0.2)node[left]{\footnotesize $\varphi(\cdot;0,\bm{\lambda}(d))$};
        \draw[blue] (4,0.6)  .. controls (2,0.8) .. (0,2)..controls (-2,0.8)..(-4,0.6)node[left]{\footnotesize $\varphi(\cdot;0,\bm\lambda(t))$};
        \node at (0,-0.2) {\footnotesize 0};
        \draw (-1.4,0)node[below]  {\footnotesize $-d$} -- (-1.4,1) -- (1.4,1) -- (1.4,0)node[below] {\footnotesize $d$};
        \node at (0.15,1.1) {\footnotesize 1};
    \end{tikzpicture}
    \caption{The Transfer $t$ and The Cutoff Transfer $d$. The given transfer $t$ and the matching cutoff transfer $d$ are shown in black. The wider distribution shown in blue is the signal distribution chosen by the agent under $t$, while the narrower (more precise) distribution in red is the one chosen under $d$. } \label{pic: Delta t}
\end{figure}

Next, I argue that 
\begin{align}
\label{eq:E-d-monotone}
    E(\lambda;d)-E(\lambda;t) \geq 0 &\text{ for } \lambda > \bm{\lambda}(t) \nonumber\\
    E(\lambda;d)-E(\lambda;t) \leq 0 &\text{ for } \lambda < \bm{\lambda}(t).
\end{align}
This is depicted in Figure~\ref{pic: illustration}.

To see this, note that as the cutoff transfer pays the entire budget $1$ when $|\theta-a|\leq d$, it is larger than the transfer $t$ within the cutoff region. In addition, Lemma \ref{lem: ratio}  shows that, since the elasticity of $\phi$ is weakly increasing,  for any $0<x_1<d<x_2$ the ratio 
$$\frac{\varphi(x_1;0,\lambda)}{\varphi(x_2;0,\lambda)}$$
increases in $\lambda$. Since $E(\lambda;t)$ is the integral of the transfer with respect to the signal distribution, \eqref{eq:E-d-monotone} follows. As a result, we have $\bm\lambda(d)\geq \bm\lambda (t)$ (see Figure \ref{pic: illustration}).
\begin{figure}[htp]
    \centering
    \begin{tikzpicture}[scale=0.8]
        \draw[thick, ->] (0,0) -- (0,6) node[left]{$E$};
        \draw[thick, ->] (0,0) -- (9,0) node[right]{ $\lambda$};
        \draw (7, 4)node[right] {$c(\lambda)$} .. controls (5,1) .. (1,0.3) ;
        \node[below,color=red] at (4.2,0) {$\bm{\lambda}(t)$};
        \draw[dotted] (4.2,0) -- (4.2,4.1);
        \draw[color=red] (1.5,0.3) .. controls (4,4.5) .. (7,5)node[right]{\footnotesize $E(\lambda; t)$} ;
        
        \draw (2,0.2) .. controls (4.5,5) .. (7,6)node[right]{\footnotesize $ E(\lambda; d)$} ;
        \fill (4.2,4.1) circle (0.05);
    \end{tikzpicture}
    \caption{\small Expected Transfer.} \label{pic: illustration}
\end{figure}
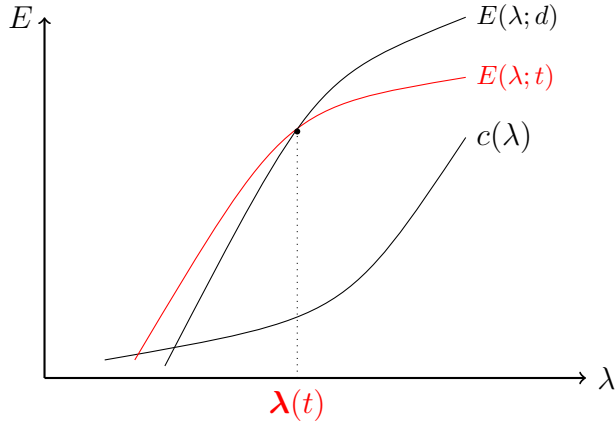


 To complete the proof of sufficiency in the general setting, there are two more technical issues. First, for a general transfer $t$ (which is not necessarily symmetric), the agent's report depends on both the signal and the precision. But in our previous example (Figure \ref{pic: illustration}), the agent's report equals the signal and does not depend on the precision. Second, the proof sketch above assumed globally increasing elasticity, which is stronger than increasing elasticity above $1$. To see why only part of the error distribution enters the condition, consider raising the reward at a small forecast error. For a nonzero error in the interior of the support,
\[
 \frac{\partial\log\varphi(x;0,\lambda)}{\partial\log\lambda}
 =1-\eta(\lambda |x|).
\]
If $|x|<\eta^{-1}(1)/\lambda^*$, the density at $x$ weakly increases as precision rises to $\lambda^*$ from any lower level. Maximizing transfer throughout this region therefore adds at least as much expected payment at $\lambda^*$ as at every lower precision, which strengthens incentives. The proof in the appendix formally deals with both issues.

\subsection{Optimization of Cutoff Transfers}
When $\phi$ satisfies increasing elasticity above $1$, Theorem \ref{thm: step} reduces the infinite dimensional optimization problem to a one-dimensional problem of finding the value of the optimal cutoff. In this section, I will solve this one-dimensional problem and identify the optimal cutoff.

As the cutoff $d$ increases, the agent's expected transfer $E(\cdot; d)$ increases. Let $\bar d$ denote the minimum cutoff such that IR constraint holds, i.e.,
$$\bar d= \min \{d\geq 0| \pi(d)\geq 0\}.$$
Notice that $\bar d$ depends on the cost function $c$. The continuity of $E(\cdot; d)$ for all $d$ and the lower semicontinuity of $c$ ensures that this minimum is well defined.

\begin{thm}\label{thm: optimal step function}
If $\phi$ satisfies increasing elasticity above $1$, the cutoff transfer $$d^*= \min \{d\geq \bar d| \bm{\lambda}(d) d\geq \eta ^{-1}(1)\}$$ is optimal.  
\end{thm}

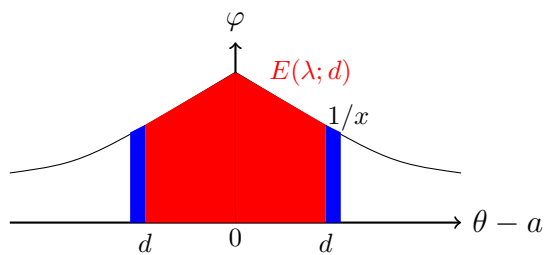
\begin{figure}[htp]
    \centering
    \begin{tikzpicture}[scale=1]
        \draw[thick, ->] (0,0) -- (0,2.4) node[above]{\small $\varphi$};
        \draw[thick, ->] (-3,0) -- (3,0) node[right]{ $\theta-a$};
       
        \draw (3,0.66)  .. controls (2,0.8) .. (0,2)..controls (-2,0.8)..(-3,0.66);
        \fill[red, opacity=0.2] (0,2) -- (1.2,1.3) -- (1.2,0)-- (0,0) -- cycle;
        \fill[blue] (1.2,1.3) --(1.4,1.2) -- (1.4,0) -- (1.2,0)-- cycle;
        \fill[red, opacity=0.2] (0,2) -- (-1.2,1.3) -- (-1.2,0)-- (0,0) -- cycle;
        \fill[blue] (-1.2,1.3) --(-1.4,1.2) -- (-1.4,0) -- (-1.2,0)-- cycle;
        \node at (0,-0.2) {\footnotesize 0};
        \node at (1.5,1.4) {\footnotesize $1/x$};
        \node[below] at (1.2,0) {\footnotesize $d$};
        \node[below] at (-1.2,0) {\footnotesize $d$};
        \node[color=red] at (1,2)  {\footnotesize $E(\lambda;d)$};
    \end{tikzpicture}
    \caption{\small Increment of Expected Transfer When Increasing $\lambda$.\\\small  The agent slightly increases precision from $\lambda$ to $\lambda+\Delta\lambda$. The area of the red region is the expected transfer $E(\lambda;d)$.
    The area of two blue regions is the increment of probability that the signal lies within the cutoff region.} \label{pic: eta }
\end{figure}

To obtain this result, the following lemma is crucial.
\begin{lem}[Complements or Substitutes]\label{lem: cross derivative} The expected transfer $E(\lambda;d)$ satisfies:
    $$\frac{\partial^2 E(\lambda; d)}{\partial \lambda\partial d}\geq 0\quad \Leftrightarrow \quad \eta (\lambda d)\leq 1.$$
\end{lem}
This lemma characterizes whether the precision and the cutoff are complements or substitutes. When the cross derivative $\frac{\partial^2 E(\lambda; d)}{\partial \lambda\partial d}$ is positive, increasing the cutoff would increase the marginal return to the precision; thus they are complements. On the other hand, when the cross derivative is negative, the cutoff and the precision are substitutes.

I illustrate the intuition of the boundary case when $\eta (\lambda d)=1$ and show that  $$\frac{\partial^2 E(\lambda; d)}{\partial \lambda\partial d}=0.$$ 
I plot the agent's the expected transfer in Figure \ref{pic: eta }. Since the agent receives transfer $1$ when the signal lies in the cutoff, the expected transfer $E(\lambda;d)$ is the probability that the signal lies in the cutoff, which is the area of the light red region. Suppose that the agent slightly increases precision from $\lambda$ to $\lambda+\Delta\lambda$. The signal shall lie in the cutoff with a higher probability. The area of dark blue region is the incremental probability that the signal lies in the cutoff. Since $\eta (\lambda d)=1$, in a neighborhood of $ d$, $\varphi(\cdot;0,\lambda )$ behaves like $x^{-1}$. We can compute the incremental expected transfer, which coincides with the incremental probability, by the area of the dark blue region 
$$E(\lambda+\Delta\lambda;d)-E(\lambda;d)\approx 2 \varphi(d;0,\lambda )\left(d-d\frac{\lambda}{\lambda+\Delta\lambda}\right)$$
which is independent of $d$ as $\varphi(\cdot;0,\lambda )$ decays at $x^{-1}$. As a result, $\frac{\partial^2 E(\lambda; d)}{\partial \lambda\partial d}=0$. A similar argument shows that $\frac{\partial^2 E(\lambda; d)}{\partial \lambda\partial d}$ and $\eta (\lambda d)-1$ have the opposite sign. Here $1$ appears as this problem is one dimensional, which is also why ``increasing elasticity above 1" appears in Theorem \ref{thm: step} and $\eta^{-1}(1)$ appears in Theorem \ref{thm: optimal step function}. Later on when I generalize this problem to $n$-dimensional, $n$ shows up in both characterizations.

\begin{figure}[htp]
    \centering
    \begin{tikzpicture}[scale=0.7]
        \draw[thick, ->] (0,0) -- (0,6) node[left]{$E$};
        \draw[thick, ->] (0,0) -- (9,0) node[right]{ $\lambda$};
        \draw (1,0.2) .. controls (4.2,4.4) .. (7,5)node[right]{\footnotesize $ E(\lambda; d)$} ;
        \draw[dashed] (0,2.5) -- (7,2.5)node[right]{\footnotesize $\lambda d=\eta ^{-1}(1)$};
        \node[below] at (5,2.5) {Complement};
        \node[below] at (5,3.5) {Substitute};
    \end{tikzpicture}
    \caption{\small Expected Transfer} \label{pic Substitute Complement}
\end{figure}
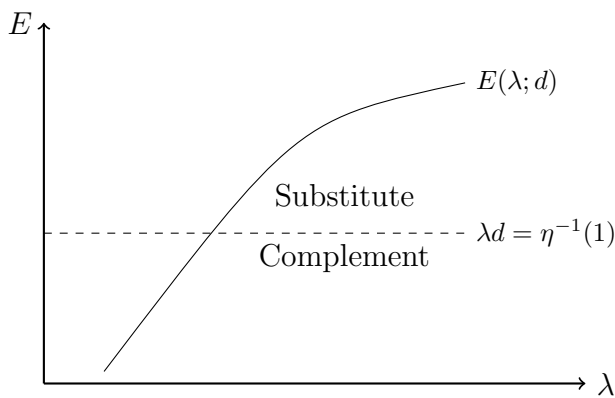

As the expected transfer $E(\lambda;d)$ is the probability that a signal lies within the cutoff region, we have $E(\lambda;d)=2\Phi(\lambda d)-1$ where $\Phi$ is the CDF of the distribution $\phi$. Consequently, $E(\lambda;d)$ is increasing in $\lambda d$. Since $\frac{\partial^2 E(\lambda; d)}{\partial \lambda\partial d}$ and $\eta (\lambda d)-1$ have the opposite sign, by the definition of $\eta ^{-1}(1)$, we have
$$\frac{\partial^2 E(\lambda; d)}{\partial \lambda\partial d}\geq 0\quad\text{if}\quad \lambda d\leq \eta ^{-1}(1)$$
$$\frac{\partial^2 E(\lambda; d)}{\partial \lambda\partial d}\leq 0\quad\text{if}\quad \lambda d\geq \eta ^{-1}(1).$$
Thus, the graph of expected transfer is separated into two parts (see Figure \ref{pic Substitute Complement}). When $ \lambda d\geq \eta ^{-1}(1)$, $\frac{\partial^2 E(\lambda; d)}{\partial \lambda\partial d}\leq 0$, I refer to this part as the \textit{substitute region}. When $ \lambda d< \eta ^{-1}(1)$, $\frac{\partial^2 E(\lambda; d)}{\partial \lambda\partial d}>0$, I refer to this part as the \textit{complement region}.

Recall that $\bar d$ denotes the minimum cutoff that the agent accepts the contract. If $(\bar d,\bm\lambda(\bar d))$ lies in the substitute region, increasing $d$ shall reduce $\bm\lambda(d)$, which is not optimal for the principal. Consequently, she should set the cutoff at the minimum $\bar d$. In this case, the IR constraint is binding and the agent's surplus is zero. If $(\bar d,\bm\lambda(\bar d))$ lies in the complement region, increasing $d$ shall first increase $\bm\lambda(d)$, and $\bm\lambda(d) d$ shall increase until it hits the boundary $\eta ^{-1}(1)$. After hitting the boundary, if we further increase $d$, then $\bm\lambda(d)$ shall decrease as $(d, \bm\lambda(d))$ enters the substitute region. Thus, the optimal cutoff is the $d$ where $\bm\lambda(d) d$ first hits the boundary $\eta ^{-1}(1)$. In this case, the IR constraint is relaxed and the agent enjoys some surplus.

Note that although optimal precision $\lambda^*$ is unique by design, the optimal cutoff transfer is not necessarily unique. To see this, recall that changing $d$ causes the curve $E(\cdot;d)$ to rotate. If the cost function $c$ has a kink and $\bm\lambda(d)$ is stuck at this kink, rotating $E(\cdot;d)$ slightly might not affect $\bm\lambda(d)$, which gives rise to multiple optimal cutoff transfers. The cutoff $d^*$ in Theorem \ref{thm: optimal step function} is the optimal cutoff transfer that provides the strongest local incentive around $\lambda^*$, i.e., it has the largest derivative $\frac{\partial E(\lambda;d)}{\partial \lambda}\big|_{\lambda=\lambda^*}$ among all optimal cutoff transfers. But this example is not generic due to the kink in the cost function. This is also why a continuously differentiable cost function helps us to get a unique optimal transfer rule.

In Appendix~\ref{sec: comparative statics}, I conduct comparative statics analysis to study how the optimal cutoff changes with the cost function.

\subsection{The Necessity of Increasing Elasticity above 1}
In this section, I explain the intuition behind the necessity result---(2) implies (3)---in Theorem~\ref{thm: step}. That is, why increasing elasticity above 1 is necessary for a cutoff transfer to always be optimal.

I first provide an example where $\phi$ does not satisfy increasing elasticity above $1$. In this case, for some increasing and differentiable cost function, all cutoff transfers are suboptimal.
\begin{exam}\label{counterexample}
The standardized signal distribution is truncated $\exp(1/x)$:
\[\phi(x)=\begin{cases}
k\exp(1/\epsilon), & \text{ if } |x|\in[0,\epsilon),\\
k\exp(1/|x|), & \text{ if } |x|\in[\epsilon,1],\\
0, & \text{ otherwise,}\end{cases}\]
where $k$ is a normalizing factor (see the left panel of Figure~\ref{fig: exam distribution}).\footnote{We can smooth this density function near $\epsilon$ and the boundary without changing the local elasticity pattern.} I plot the elasticity of $\phi$ on the right panel of Figure \ref{fig: exam distribution}. The signal distribution does not satisfy increasing elasticity above $1$ as $\eta(x)$ is decreasing for $x\in[\epsilon,1]$. Fix a pair $(\lambda^*,d^*)$ such that $\epsilon<\lambda^* d^*<1$. Suppose an increasing cost function $c\in C^1$ is tangent to $E(\cdot;d^*)$ at $\lambda=\lambda^*$ and $c(\lambda)$ is strictly above $E(\lambda;d^*)$ for all $\lambda\neq \lambda^*$ (see the right panel of Figure \ref{fig:exam transfer}).
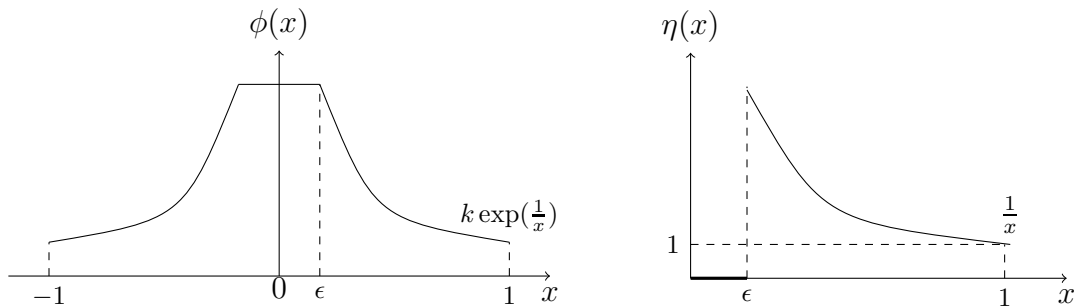
\begin{figure}[!htb]
    \begin{minipage}[b]{0.5\linewidth}
    \centering
    \begin{tikzpicture}[xscale=1.8, yscale=1.5]
    \draw[->](-2,0)--(2,0)node[left,below]{$x$};
    \draw[->](0,0)--(0,2)node[above]{$\phi(x)$};
    \draw (-0.3,1.7)--(0.3,1.7) ;
    \node (O) at (0,-0.1) {0};
    \draw  (0.3,1.7)..controls (0.7,0.5)..(1.7,0.3)node[above]{\footnotesize $k\exp(\frac{1}{x})$};
    \draw  (-0.3,1.7)..controls (-0.7,0.5)..(-1.7,0.3);
    \draw[dashed] (1.7,0)node[below]{\footnotesize $1$} -- (1.7,0.3);
    \draw[dashed] (-1.7,0)node[below]{\footnotesize $-1$} -- (-1.7,0.3);
    \draw[dashed] (0.3,0)node[below]{\footnotesize $\epsilon$} -- (0.3,1.7);
    \end{tikzpicture}
     \end{minipage}
    \begin{minipage}[b]{0.5\linewidth}
     \centering
    \begin{tikzpicture}[xscale=2.5,yscale=1.5]
    \draw[->](0,0)--(2,0)node[left,below]{$x$};
    \draw[->](0,0)--(0,2)node[above]{$\eta(x)$};
    \draw  (0.3,1.67)..controls (0.7,0.5)..(1.7,0.3)node[above]{$\frac{1}{x}$};
    \draw[dashed] (1.67,0)node[below]{\footnotesize $1$} -- (1.67,0.3);
    \draw[dashed] (0.3,0)node[below]{\footnotesize $\epsilon$} -- (0.3,1.7);
    \draw[dashed] (0,0.3)node[left]{\footnotesize $1$} -- (1.67,0.3);
    \draw[very thick] (0,0)--(0.3,0) ;
    \end{tikzpicture}
     \end{minipage}
      \caption{Signal Distribution and Elasticity}
    \label{fig: exam distribution}
\end{figure}

\begin{figure}[!htb]
    \begin{minipage}[b]{0.5\linewidth}
     \centering
    \begin{tikzpicture}[scale=1.8]
    \draw[thick, ->](-2.1,0)--(2.1,0)node[left,below]{$s$};
    \draw[thick, ->](0,0)node[below]{$\theta$}--(0,2);
    \draw (-0.3,1.7)--(0.3,1.7) ;
    \draw[red] (-1,0)--(-1,1.3)--(-0.6,1.3)--(-0.6,0)--(-0.5,0)--(-0.5,1.3)--(0.5,1.3)--(0.5,0)--(0.6,0)--(0.6,1.3)--(1,1.3)--(1,0);
    
    \draw[red] (-1.6,0)--(-1.6,1.3)--(-1.4,1.3)--(-1.4,0);
    \draw[red] (1.6,0)--(1.6,1.3)node[above]{$\underline t$}--(1.4,1.3)--(1.4,0);
    \fill[blue, opacity=0.6] (0.5,0)--(0.5,1.14)--(0.6,0.9)--(0.6,0) -- cycle;
    \fill[blue, opacity=0.6] (-0.5,0)--(-0.5,1.14)--(-0.6,0.9)--(-0.6,0) -- cycle;
    \fill[brown, opacity=0.2] (1.6,0)--(1.6,0.4)--(1.4,0.43)--(1.4,0) -- cycle;
    \fill[brown, opacity=0.2] (-1.6,0)--(-1.6,0.4)--(-1.4,0.43)--(-1.4,0) -- cycle;
    \draw  (0.3,1.7)..controls (0.7,0.5)..(2.1,0.3);
    \node at (0.5,1.8) {\footnotesize $\varphi(s;\theta,\lambda^*)$};
    \draw  (-0.3,1.7)..controls (-0.7,0.5)..(-2.1,0.3);
    \node[below] at (1,0) {\footnotesize $d^*$};
    \end{tikzpicture}

     \end{minipage}
     \begin{minipage}[b]{0.5\linewidth}
    \centering
    \begin{tikzpicture}[scale=0.6]
    \draw[thick, ->] (0.5,0) -- (0.5,6) node[left]{$E$};
        \draw[thick, ->] (0.5,0) -- (9,0) node[right]{ $\lambda$};
        \draw (7, 6)node[right] {$c(\lambda)$} .. controls (5,3) .. (1,2.3) ;
        \node[below,color=red] at (4.5,0) {$\lambda^*$};
        \draw[dotted] (4.5,0) -- (4.5,3.1);
        \draw (1.5,0) .. controls (3.8,3.04) .. (7,4)node[right]{\footnotesize $E(\lambda; d^*)$} ; 
        \draw[color=red] (2.3,-0.2) .. controls (4.89,4) .. (7.2,5)node[right]{\footnotesize $ E(\lambda; \underline t)$} ;
        \fill (4.5,3.15) circle (0.05);
        \draw[dashed, color=blue] (0.5,2)--(9,2); 
        \node[below, color=blue] at (7,1.8) {\footnotesize Complement};
        \node[below, color=blue] at (7,3) {\footnotesize Substitute};
    \end{tikzpicture}
     \end{minipage}
      \caption{The New Transfer Rule and The Expected Transfer}
    \label{fig:exam transfer}
\end{figure}
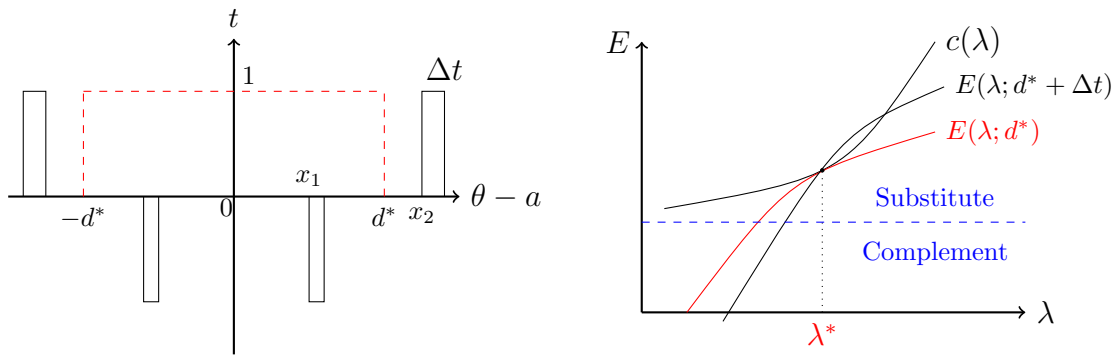
First, note that $d^*$ is the best cutoff among all cutoff transfers: by design, $d^*$ is the minimum cutoff for the agent to work. By Lemma \ref{lem: cross derivative}, as $(\lambda^*,d^*)$ lies in the substitute region, any cutoff transfer with cutoff $d>d^*$ induces a smaller precision. Thus, $d^*$ is the best cutoff transfer.

Second, I construct a new transfer rule $\underline t$ which induces a strictly larger precision. Starting with the best cutoff transfer rule $d^*$, I modify the transfer rule by setting one interior region (dark blue) within the cutoff to 0 and one exterior region (light brown) outside the cutoff to 1 (see the left panel of Figure \ref{fig:exam transfer}). These two regions are chosen such that the area of the dark blue region and the light brown region are the same. Note that the area of each region is the expected transfer contributed by this region at precision $\lambda^*$. Thus, the new transfer rule $\underline t$ has the same expected transfer as cutoff transfer $d^*$ at precision $\lambda^*$.

As the signal distribution features decreasing elasticity, the density function $\varphi(s;\theta,\lambda^*)$ is steeper at the blue region. Once we increase precision, the blue area shall shrink more than the brown area. This implies that the new transfer rule $\underline t$ has higher expected transfer than the cutoff transfer $d^*$ for slightly higher precision. I translate this comparison to the right panel of Figure \ref{fig:exam transfer}. For higher precision $\lambda>\lambda^*$, $E(\lambda;\underline t)>E(\lambda;d^*)$. Thus, $\underline t$ shall induce a higher precision.

\end{exam}

This example is key to understand more generally the necessity of increasing elasticity above $1$. I call a pair $(\lambda^*, d^*)$ \textit{exposed} if the cutoff transfer $d^*$ is the optimal transfer for some increasing convex cost function $c\in C^1$ and induces precision $\lambda^*$. A pair $(\lambda^*, d^*)$ being exposed implies that I cannot perturb the optimal cutoff transfer $d^*$ to generate a larger incentive for the agent to increase the precision. That is,
$$\frac{\partial E(\lambda; d^*)}{\partial \lambda}\bigg|_{\lambda=\lambda^*}$$
cannot be increased by any perturbations of cutoff transfer $d^*$. Otherwise, any small increment shall be detected by the continuously differentiable cost function $c$ and reflected in a larger optimal precision $\arg\max_{\lambda} E(\lambda;t)-c(\lambda)$. This explains why only smooth cost functions are needed as testing functions.

There are at least two ways to perturb the cutoff transfer $d^*$
\begin{itemize}
    \item First, I can slightly increase the cutoff (without affecting IR). This implies that $(\lambda^*, d^*)$ must be substitutes,
    $$\frac{\partial^2 E(\lambda; d)}{\partial \lambda\partial d}\bigg|_{(\lambda^*,d^*)}\leq 0\quad \Leftrightarrow \quad \eta (\lambda^* d^*)\geq 1.$$
    Otherwise, increasing the cutoff shall lead to a precision larger than $\lambda^*$.
    \item Second, I can alter the cutoff transfer $d^*$ as in Figure \ref{fig:exam transfer}. This implies that $\eta (x_1)\leq \eta (x_2)$ for all $0\leq x_1\leq \lambda^* d^*\leq x_2$. Otherwise, I can use the construction in Example \ref{counterexample} to construct a new transfer that induces a strictly larger precision.
\end{itemize}
When for every increasing, convex, continuously differentiable cost function, there exists an optimal transfer that is a cutoff transfer. Exposing any pair $(\lambda^*,d^*)$ with  $\eta (d^*\lambda^*)\geq 1$ is not hard, as I can always find an increasing, convex, continuously differentiable cost function that is tangent to $E(\cdot;d^*)$ at $\lambda^*$ and strictly above $E(\cdot;d^*)$ everywhere else, like in Figure \ref{fig:exam transfer}. This implies that all pairs $(\lambda^*,d^*)$ with $\eta (d^*\lambda^*)\geq 1$ are exposed. Then the second bullet point above implies increasing elasticity above 1. The first bullet point above implies that all pairs in the complement regions are not exposed. Consequently, we do not have  monotonicity for $\eta$ when it is below 1.

\section{Extensions}\label{sec: extensions}
\subsection{The Multi-Dimensional Case}\label{sec: multidimensional}
My results can generalize to the case with a multi-dimensional state. In this section, suppose the state, the signal, and the report are n-dimensional, i.e., $\theta,\,s,\, a\in \mathbb R^n$. The signal distribution is still symmetric, single-peaked, and centered around $\theta$. Here symmetry means that the density of the distribution depends only on the Euclidean distance from the state $\theta$. A cutoff transfer pays $1$ when the Euclidean distance between the report and the state is less than a cutoff $d$. The next proposition shows that a statement analogous to that of Theorem~\ref{thm: step} applies in this case too.
\begin{prop}\label{prop: high dimensional}
Suppose the state is n-dimensional. The following statements are equivalent.
    \begin{enumerate}
        \item For every cost function, there exists an optimal transfer that is a cutoff transfer.
        \item For every increasing, convex, and continuously differentiable cost function, there exists an optimal transfer that is a cutoff transfer.
        \item Density function $\phi$ satisfies increasing elasticity above $n$.
    \end{enumerate} 
Moreover, if $\phi$ satisfies increasing elasticity above $n$, the cutoff transfer$$d^*= \min \{d\geq \bar d| \bm{\lambda}(d) d\geq \eta ^{-1}(n)\}$$ is optimal. 
\end{prop}
Note that the increasing elasticity condition of Theorem~\ref{thm: step} is generalized in this proposition to increasing elasticity above the dimension $n$, providing an explanation for why the number 1 appeared in Theorem~\ref{thm: step}.

\subsection{Gaussian Prior and Gaussian Signal}\label{sec: Gaussian prior}

In this section, I assume that both the prior and the signal admit Gaussian distributions. Since the prior is no longer uniform, a slight adjustment is that now the principal wants the agent to truthfully report the posterior mean rather than the signal.

Let the prior distribution be $\mathcal N( 0, 1/\lambda_0^2)$.  Conditional on a signal $s$, the agent's posterior is also Gaussian
$$\mathcal N(\frac{s\lambda^2}{\lambda_0^2+\lambda^2}, \frac{1}{\lambda_0^2+\lambda^2}),$$
where $\lambda$ is the agent's choice of precision. I let $\Lambda$ denote the precision of the posterior
$$\Lambda(\lambda)= \sqrt{\lambda_0^2+\lambda^2}.$$
Then, the expected transfer becomes $E(\Lambda(\lambda);t)$ instead of $E(\lambda; t)$. Therefore, $\Lambda(\lambda)$ shall replace $\lambda$ in the characterization. Other than this difference, the results remain the same as in previous sections, as the posterior is Gaussian whose elasticity is increasing.
\begin{prop}\label{prop: Gaussian}
Suppose the prior and the signal admit  Gaussian distribution. There exists an optimal transfer that is a cutoff transfer. Moreover, the cutoff transfer $$d^*= \min \{d\geq \bar d| \Lambda(\bm{\lambda}(d)) d\geq \eta ^{-1}(1)\}$$ is optimal. 
\end{prop}

\subsection{Unobserved State}\label{sec: unobserved state}
In this section, I extend my results to the setting where the state is unobservable. Instead, the principal has some private information about the state. When designing the transfer, she uses her private information instead of the state to discipline the agent.\footnote{Note that the principal does not have utility for their budget, and so have no incentive to lie about their private signals. Thus, it is without loss of generality to assume that the principal has commitment power.} 

I assume that the prior is uniform or Gaussian $\mathcal N( 0,1/\lambda_0^2)$. The principal wants the agent to truthfully report his posterior mean and maximizes $\lambda$.  The principal and the agent receive independent Gaussian signals: $s_p\sim \mathcal N(\theta,1/\lambda_p^2)$ and $s\sim \mathcal N(\theta,1/\lambda^2)$. I show that it is without loss to focus on cutoff transfers that pay $1$ when $|s_p-a|$ is less than a cutoff and $0$ otherwise.

\begin{prop}\label{prop: unobserved}
Suppose that the state is unobservable and the prior admits uniform or Gaussian distribution. The principal and the agent receive independent Gaussian signals. There exists an optimal transfer that is a cutoff transfer. Moreover, the cutoff transfer $$d^*= \min \{d\geq \bar d| \Lambda(\bm{\lambda}(d)) d\geq \eta ^{-1}(1)\}$$ is optimal where $\Lambda(\bm{\lambda}(d))= (\frac{1}{\lambda_p^2}+\frac{1}{\bm\lambda^2(d)})^{-\frac{1}{2}}$ for the uniform prior and $\Lambda(\bm{\lambda}(d))= (\frac{1}{\lambda_p^2}+\frac{1}{\bm\lambda^2(d)+\lambda_0^2})^{-\frac{1}{2}}$ for the Gaussian prior. 
\end{prop}

\section{When the Principal Values Savings}\label{sec: savings}

So far, I have assumed that the principal does not value the unused budget. What happens if the principal also cares about savings? As a robustness check, I return to the baseline model and allow the principal to trade off the benefit of precision against the expected transfer, while retaining the hard budget constraint. The principal thereby maximizes the gross benefit minus the expected transfer and wants to induce truthful report. Let $U(\lambda)$ denote the principal's gross benefit from precision, where $U$ is increasing. The agent is risk neutral, so that transfers are measured in money. The signal distribution and the class of admissible transfers are unchanged.   As in Section~3, I first relax the requirement that the transfer induce truthful reporting:
\begin{maxi}
{\lambda, t}
{U(\lambda)-E(\lambda;t)}
{\label{eq: savings program}}
{}
\addConstraint{0}{\le t\le 1}
\addConstraint{\lambda}{\in \arg\max E(\cdot;t)-c(\cdot)}{\quad\text{IC}}
\addConstraint{E}{(\lambda;t)-c(\lambda)\ge 0}{\quad\text{IR}}
\end{maxi}

This problem can be decomposed into a two-step optimization problem. The first step is to minimize the expected transfer for each fixed target $\bar\lambda$:
\begin{mini}
{t}
{E(\bar\lambda;t)}
{\label{eq: cost minimization}}
{}
\addConstraint{0}{\le t\le 1}
\addConstraint{\bar\lambda}{\in \arg\max E(\cdot;t)-c(\cdot)}{\quad\text{IC}}
\addConstraint{E}{(\bar\lambda;t)-c(\bar\lambda)\ge 0}{\quad\text{IR}}
\end{mini}
The second step is to choose precision to maximize the gross benefit minus this minimum expected transfer. As before, a transfer is optimal for a given target if it solves the relaxed problem~\eqref{eq: cost minimization} and induces truthful reporting. I call a target \textit{implementable} if this problem has a feasible transfer.\footnote{Implementation in this section means that the target belongs to the agent's set of best responses, as specified by IC. It need not be the largest best response selected by $\bm\lambda(t)$ in Section~4. The cutoff characterized below uniquely implements a positive target with positive marginal cost if $\phi$ is strictly decreasing on the positive half of its support, or if $c$ is strictly convex.} Let $C(\bar\lambda)$ denote the infimum in~\eqref{eq: cost minimization}, with $C(\bar\lambda)=+\infty$ if the target is not implementable.

I take $\lambda\in[0,\infty)$. Different from the baseline model, I impose the following restriction on the cost function.
\begin{asmp}\label{ass: savings cost}
The cost function belongs to
\[
\C=\left\{
c\in C^1([0,\infty)):
c(0)=0,\quad c\text{ is nondecreasing and convex}
\right\}.
\]
\end{asmp}
Here, $c(0)=0$ ensures that participation follows from the costless zero-precision option.

\begin{defn}\label{def: increasing elasticity below 1}
The signal distribution satisfies \textit{increasing elasticity below 1} if
\[
    x\longmapsto\min\{\eta(x),1\}
\]
is weakly increasing.
\end{defn}
Thus, elasticity is weakly increasing until it first reaches $1$ and never subsequently falls below $1$. The condition imposes no monotonicity restriction on elasticity once it remains above $1$. The following theorem is a counterpart of Theorem~\ref{thm: step} for a principal who minimizes the cost of implementing a given precision.
\begin{thm}\label{thm: cost minimizing cutoff}
The following are equivalent.
\begin{enumerate}
    \item For every cost function $c\in \C$ and every implementable target $\bar\lambda$, there exists an optimal transfer that is a cutoff transfer.
    \item For every increasing linear cost function $c(\lambda)=k\lambda$ with $k>0$ and every implementable target $\bar\lambda$, there exists an optimal transfer that is a cutoff transfer.
    \item Signal density function $\phi$ satisfies increasing elasticity below $1$.
\end{enumerate}
\end{thm}

The economic force behind this result is the marginal incentive provided by each expected dollar of payment. At an error $x$ in the positive density support, the ratio of the change in density to the density itself is
\begin{equation}\label{eq: savings score}
    \left.
    \frac{\partial\varphi(x;0,\lambda)/\partial\lambda}
         {\varphi(x;0,\lambda)}
    \right|_{\lambda=\bar\lambda}
    =\frac{1-\eta(\bar\lambda|x|)}{\bar\lambda}.
\end{equation}
A small increase in the payment near $x$ raises the expected transfer in proportion to the denominator and raises the marginal return to precision in proportion to the numerator. Their ratio therefore measures how efficiently a payment rewards precision. Increasing elasticity below $1$ says that the positive part of this ratio is weakly decreasing in the absolute error. Payments on errors closer to zero provide at least as much positive marginal incentive per expected dollar. With a unit payment cap, the principal can minimize expenditure by paying the entire prize on a central interval.

This also explains why only elasticity below $1$ matters. When $\eta(\bar\lambda|x|)\geq1$, payment near $x$ does not provide a positive marginal incentive for precision. How these nonpositive contributions are ordered is immaterial for the least expensive way to provide a positive incentive. Conversely, if increasing elasticity below $1$ fails, a more distant error can provide a strictly higher positive marginal incentive per expected dollar than a closer error. A small reallocation of payments can then preserve the incentive at the target while reducing expenditure. The proof shows that such a reallocation preserves global feasibility for a suitable linear cost function. Thus, the necessity result already holds within the simplest class of convex costs.

The sufficiency proof proceeds in two steps. First, a necessary
marginal first-order condition yields a lower bound on the expected
payment of \textit{any} transfer implementing the target. Increasing elasticity
below $1$ orders positive marginal incentives per expected dollar. Then I verify a cutoff transfer can attain this bound by concentrating payments
on errors closest to zero. Second, I verify that this cutoff satisfies
the global incentive compatibility constraint. Thus, the proof combines a payment
lower bound for arbitrary feasible transfers with a separate
verification of global feasibility for the candidate cutoff.

Compared with Theorem~\ref{thm: step}, this characterization uses three additional restrictions. First, risk neutrality is now substantive. In the precision-maximization problem, I can express transfers in utils without changing the principal's objective. When she values monetary savings, a nonlinear utility function $u$ would instead make the monetary cost of a utility transfer $u^{-1}(t)$. The same normalization would therefore change her objective.

Second, convexity of costs ensures that the marginal condition for the proposed cutoff is sufficient for global incentive compatibility. For any cutoff $d$, function $E(\cdot;d)$ is concave. Consequently, $E(\cdot;d)-c(\cdot)$ is concave under Assumption~\ref{ass: savings cost}. Thus, I do not need to assume the validity of the first-order approach \citep{rogerson1985first,jewitt1988justifying}. Unlike Theorem~\ref{thm: step}, this result does not cover all lower-semicontinuous cost functions.

Third, $c(0)=0$ ensures that participation follows from the costless zero-precision option. In Theorem~\ref{thm: step}, the cost function space is strictly larger and does not impose  $c(0)=0$. Alternatively, we could also impose $c(0)=0$ while allowing for arbitrary outside option utility $\bar u\geq 0$ in Theorem~\ref{thm: step}. Then the restriction in Theorem~\ref{thm: cost minimizing cutoff} is that the outside option utility is zero.

\begin{figure}[htbp]
    \centering
    \includegraphics[width=0.5\textwidth]{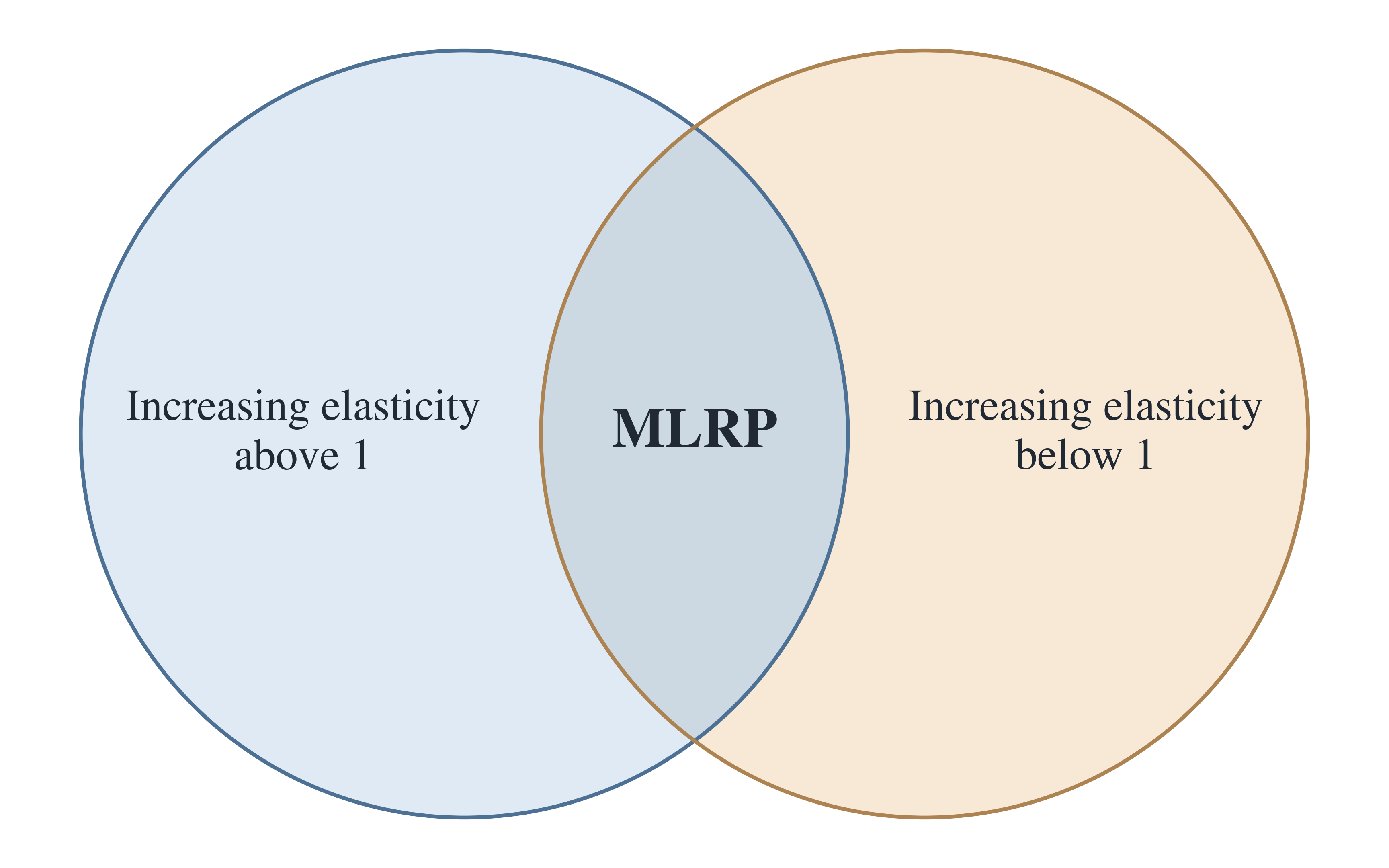}
    \caption{Relationship between the elasticity conditions.
    The intersection of increasing elasticity above $1$ and increasing
    elasticity below $1$ consists exactly of signal distributions
    satisfying MLRP.}
    \label{fig:elasticity_venn}
\end{figure}
Finally, the two objectives, maximizing precision and minimizing expected transfer, impose complementary restrictions on the signal distribution. Increasing elasticity above $1$ is equivalent to weak monotonicity of $\max\{\eta,1\}$, whereas increasing elasticity below $1$ is weak monotonicity of $\min\{\eta,1\}$. Their intersection is exactly global weak monotonicity of $\eta$, since
\[
    \eta=\min\{\eta,1\}+\max\{\eta,1\}-1.
\]
By Lemma~\ref{lem: ratio}, global increasing elasticity is equivalent to MLRP. Hence MLRP is precisely the intersection of the distributional conditions for universal cutoff optimality in the two problems (Figure~\ref{fig:elasticity_venn}).

\subsection{The Optimal Cutoff}
Once the condition in Theorem~\ref{thm: cost minimizing cutoff} holds, the remaining problem is to determine the optimal cutoff. Define
\begin{equation}\label{eq: savings GH}
    G(z)=2\Phi(z)-1,
    \qquad
    H(z)=2z\phi(z),
    \qquad z\geq0.
\end{equation}
The function $G$ gives the expected transfer under a cutoff with standardized radius $z$, while $H$ measures its marginal incentive for precision. In particular,
\[
    E(\lambda;d)=G(\lambda d),
    \qquad
    \frac{\partial E(\lambda;d)}{\partial\lambda}
    =\frac{H(\lambda d)}{\lambda},
    \qquad \lambda>0.
\]
Define $z_{\bar\lambda}$ as
\begin{equation}\label{eq: savings root}
   z_{\bar\lambda}=\inf\{z\geq 0|H(z)=\bar\lambda c'(\bar\lambda)\} 
\end{equation}

\begin{thm}\label{thm: minimum expected transfer}
Suppose that $\phi$ satisfies increasing elasticity below $1$. For any cost function $c\in \C$, a target $\bar\lambda$ is implementable if and only if
\begin{equation}\label{eq: savings feasibility}
    \bar\lambda c'(\bar\lambda)\leq\max_{z\geq0}H(z).
\end{equation}
When this condition holds, the optimal cutoff and the minimum expected transfer are
\begin{equation}\label{eq: savings payment}
    d^*(\bar\lambda)=\frac{z_{\bar\lambda}}{\bar\lambda},
    \qquad
    C(\bar\lambda)=G(z_{\bar\lambda}).
\end{equation}
\end{thm}

The feasibility condition compares the marginal incentive required at the target with the largest marginal incentive that the budget can provide. Under increasing elasticity below $1$,
\[
    H'(z)=2\phi(z)[1-\eta(z)]
\]
on the positive density support. Hence $H$ is strictly increasing until elasticity first reaches $1$ and weakly decreasing afterward. The quantity $\max H$ is the resulting limit on scaled marginal incentives. The proof shows that no admissible transfer can exceed this limit, so the feasibility condition applies to all transfers.

Several cutoffs may provide the same marginal incentive. The smallest solution in~\eqref{eq: savings root} minimizes the probability of payment, since $G$ is increasing on the positive support. Theorem~\ref{thm: minimum expected transfer} shows that the resulting cutoff minimizes expected payments among all admissible transfers, not just among cutoff transfers.

The second step of the principal's problem is therefore
\[
    \sup_{\lambda\geq0:\,\lambda c'(\lambda)\leq\max_{z\geq0}H(z)}
    \{U(\lambda)-C(\lambda)\}.
\]
 Thus, whenever increasing elasticity below $1$ holds, the original contract design problem reduces to a scalar choice of precision.

\newpage

\bibliography{bibfile}
\newpage

\appendix
\section{Comparative Statics}\label{sec: comparative statics}

In this section, I study how the optimal cutoff changes with the cost function. Without loss of generality, I assume that the cost function is weakly increasing.\footnote{Given any cost function $\tilde c$ that is not weakly increasing, I can define a new cost function $c$ to be the largest weakly increasing function that is below $\tilde c$. The agent's precision choice problem under cost function $c$ is the same as under cost $\tilde c$, because $E(\lambda;d)$ is weakly increasing in $\lambda$.} I express dependence on the cost function. Let $d^*(c)$ and $\lambda^*(c)$ denote the optimal cutoff and precision given cost function $c$. Consider two cost functions $c_1\leq c_2$ where $c_1$ is less costly. Intuitively, if the agent's cost is higher, the optimal cutoff increases. The next result formalizes this intuition.

\begin{prop}[Comparative Statics]\label{prop: cost}
Suppose that $\phi$ satisfies increasing elasticity above 1 and $c_1\leq c_2$. If $c_2(\lambda)-c_1(\lambda)$ is weakly increasing in $\lambda$, then $d^*(c_1)\leq d^*(c_2)$ and $\lambda^*(c_2)\leq \lambda^*(c_1)$. In particular, if $c_2=k c_1$ for some constant $k>1$, then $d^*(c_1)\leq d^*(c_2)$ and $\lambda^*(c_2)\leq \lambda^*(c_1)$.
\end{prop}

The cost difference $c_2(\lambda)-c_1(\lambda)$ being weakly increasing occurs if $c_2$ is more convex than $c_1$. Another interesting comparative statics is to change the budget from $1$ to $1/k$. This is equivalent to keeping the budget at $1$ and change the cost function from $c_1$ to $kc_1$. Thus, lowering the budget would lead to a larger optimal cutoff.

As a direct corollary, I can study what happens if the agent's signal is more precise. Suppose that $\varepsilon_2=k \varepsilon_1$ with some constant $k>1$. Recall that the agent's signal is $s=\theta+\frac{1}{\lambda}\varepsilon$. Thus, a smaller noise $\varepsilon_1$ corresponds to a more precise signal, which leads to a smaller cutoff by the next result.
\begin{col}\label{col: comparative signal}
    Suppose that the cost function is weakly convex. If $\varepsilon_2=k \varepsilon_1$ with some constant $k>1$, $d^*_1\leq d^*_2$.
\end{col}

What if the cost difference $c_2-c_1$ is not increasing? Then the answer is more involved. Denote by $\bm\lambda(d;c)$ the induced precision given cutoff transfer $d$ and the cost function $c$.  Given the cost function $c_1$, the minimum cutoff $\bar d(c_1)$ and the induced precision $\bm\lambda(\bar d(c_1);c_1)$ pair lies either in the substitute region $\bm\lambda(\bar d(c_1);c_1) \bar d(c_1)\geq \eta ^{-1}(1)$ or not. If this pair lies in the substitute region, then the optimal cutoff coincides with the minimum cutoff, $d^*(c_1)=\bar d(c_1)$, by Theorem \ref{thm: optimal step function}. Then given a larger cost function $c_2$, the minimum cutoff must be larger, $\bar d(c_2)\geq \bar d(c_1)$, which entails a larger optimal cutoff, $d^*(c_2)\geq d^*(c_1)$.

 However, if the $\bm\lambda(\bar d(c_1);c_1),\, \bar d(c_1)$ pair lies in the complement region, the comparison of $d^*$ can go both directions. If $c_2-c_1$ is decreasing, it may render higher precision relatively more attractive to the agent under cost $c_2$ than $c_1$. This may induce a higher precision for each cutoff transfer. As the optimal precision and optimal cutoff are substitutes by Theorem \ref{thm: optimal step function}, a higher induced precision leads to a lower optimal cutoff.

Readers may wonder what happens to the optimal precision $\lambda^*$ given $c_2\geq c_1$. It turns out that increasing the cost function can shift $\lambda^*$ in both directions, depending on the shape of the cost difference $c_2-c_1$.

\section{The Assumption That Transfer Depends Only on the Difference Between the State and Report}\label{sec: proper uniform}
In this section, I show that given the uniform prior, it is without loss to assume that the transfer depends only on the difference $a-\theta$. To make the expected transfer well-defined for the uniform prior, I take two approaches. In the first approach, I define the agent's expected transfer to be conditional on the event that the signal lies in some set of finite Lebesgue measure. Then I show that the conclusion holds for any such set. The conditioning on some set is merely an artifact to deal with the improperness of unbounded uniform prior. In the second approach, I study an alternative setting where the state space is a circle in $\mathbb R^2$ and the prior is proper and uniform. This unit circle is similar to the one-point compactification of $\mathbb R$.

Throughout this section,  I allow the transfer to depend on both the state and the report. To ensure that the agent's report is well-defined, I assume that the family of maps $a \mapsto  t(y+a,a)$ parametrized by $y$ is equicontinuous. This is trivially satisfied when the transfer depends only on $\theta-a$, as $t(y+a,a)$ stays constant as $a$ varies. When the prior is improper uniform on $\mathbb R$, I additionally assume that $t(\theta,\theta+x)$ vanishes uniformly (in $\theta$) as $|x|$ tends to infinity.

First, I show that the agent's report is well-defined.

\begin{lem}\label{lem: maximum well defined}
    The maximum $$
    \max_{a \in \mathbb{R}} \int_{\mathbb R} t(\theta', a) \varphi(\theta'; s,\lambda) d\theta'
$$  is well-defined.
\end{lem}
\begin{proof}[Proof of Lemma \ref{lem: maximum well defined}]
     Fix a signal $s$ and precision $\lambda$. Let $T(a)$ be the agent's expected transfer after reporting $a$, $$T(a)= \int t(\theta',a )\varphi (\theta';s,\lambda)d\theta'.$$
 If $T(a)=0$ for all $a\in\mathbb R$, then the maximum is well-defined. Suppose that there exists $ a_1\in \mathbb R$ such that $T(a_1)>0$. As $t(\theta,\theta+x)$ vanishes uniformly (in $\theta$) as $x$ tends to infinity, there exists a $K_1>0$ such that for all $|\theta-a|>K_1$, $t(\theta,a)< \frac{T(a_1)}{2}$. Moreover, there exists a $K_2>0$ such that $\Phi(\lambda K_2)> 1-\frac{T(a_1)}{4}$. Let $K=\max(K_1,K_2)$. If $|a-s|>2K$,
\[
\begin{split}
    T(a)=& \int t(\theta',a )\varphi (\theta';s,\lambda)d\theta'\\
    =&  \int_{\theta':|\theta'-s|<K} t(\theta',a )\varphi (\theta';s,\lambda)d\theta'+\int_{\theta':|\theta'-s|\geq K} t(\theta',a )\varphi (\theta';s,\lambda)d\theta'\\
    \leq & \frac{T(a_1)}{2} \int_{\theta':|\theta'-s|<K} \varphi (\theta';s,\lambda)d\theta'+1\cdot \int_{\theta':|\theta'-s|\geq K} \varphi (\theta';s,\lambda)d\theta'\\
    \leq  & \frac{T(a_1)}{2} +2[1- \Phi(\lambda K)]\\
    =& T(a_1).
\end{split}
\]
Thus, $$\max_{a\in\mathbb R} T(a)=\max_{|a-s|\leq 2K} T(a).$$
Next, I show that $T(\cdot)$ is continuous. 
\[
\begin{split}
    T(a+\epsilon)-T(a)=&\int t(\theta',a+\epsilon )\varphi (\theta';s,\lambda)d\theta'-\int t(\theta',a )\varphi (\theta';s,\lambda)d\theta'\\
    =&\int [t(\theta',a+\epsilon )-t(\theta'-\epsilon,a )]\varphi (\theta';s,\lambda)d\theta'+\int [t(\theta'-\epsilon,a )-t(\theta',a )]\varphi (\theta';s,\lambda)d\theta'\\
    =&\int [t(\theta',a+\epsilon )-t(\theta'-\epsilon,a )]\varphi (\theta';s,\lambda)d\theta'+\int t(\theta',a )[\varphi (\theta'+\epsilon;s,\lambda)-\varphi (\theta';s,\lambda)]d\theta'
\end{split}
\]
\[
\begin{split}
    |T(a+\epsilon)-T(a)|&\leq \sup_{\theta'}| t(\theta',a+\epsilon )-t(\theta'-\epsilon,a )|+\bigg|\int t(\theta',a )[\varphi (\theta';s-\epsilon,\lambda)-\varphi (\theta';s,\lambda)]d\theta'\bigg|\\
    &\leq \sup_{\theta'}| t(\theta',a+\epsilon )-t(\theta'-\epsilon,a )|+\int |\varphi (\theta';s-\epsilon,\lambda)-\varphi (\theta';s,\lambda)|d\theta'\\
    &\leq \sup_{\theta'}| t(\theta',a+\epsilon )-t(\theta'-\epsilon,a )|+2 \int_{\theta'\leq s-\frac{\epsilon}{2}} \varphi (\theta';s-\epsilon,\lambda)-\varphi (\theta';s,\lambda)d\theta'\\
    &\leq \sup_{\theta'}| t(\theta',a+\epsilon )-t(\theta'-\epsilon,a )|+2[\Phi(\frac{\epsilon}{2}\lambda)-\Phi(-\frac{\epsilon}{2}\lambda)]\\
\end{split}
\]
Note that $\sup_{\theta'} |t(\theta',a+\epsilon )-t(\theta'-\epsilon,a )|$ tends to $0$ as $\epsilon$ tends to $0$ by the equicontinuity of $t$. Moreover, the CDF $\Phi$ is continuous. Thus, $T$ is continuous.  A continuous function achieves a maximum on a compact set. So the maximum  is well-defined.
\end{proof}

Next, I show that it is without loss to assume that the transfer depends only on the difference $\theta-a$. For a transfer rule $t(\theta-a)$ that depends only on the difference $\theta-a$, let $\hat E(\lambda;t)$ denote the agent's expected transfer when he has to report truthfully,
$$\hat E(\lambda; t)=\int_{\mathbb R}  t(\theta') \varphi (\theta'; 0, \lambda)d\theta'.$$
Let $\bm{\hat\lambda}(t)$ denote the corresponding induced precision,
\[\bm{\hat\lambda}(t)=\begin{cases}
\max [\arg\max \hat E(\cdot;t)-c(\cdot)], & \text{ if } \max \hat E(\cdot;t)-c(\cdot)\geq 0,\\
0, & \text{ otherwise.}\end{cases}\]

For a transfer rule $t(\theta,a)$ that depends on both the state and the report, we need to integrate the signal when computing the expected transfer. Note that the signal distribution is also uniform. Fix any measurable set $A$ with $0<|A|<\infty$. I define the agent's expected transfer conditional on the signal lies in set 
$A$, i.e., $s\in A$
$$E(\lambda;t)=\int_{s\in A} \frac{1}{|A|}ds\bigg(\max_a \int t(\theta', a)\varphi(\theta';s,\lambda)d\theta'\bigg)$$
\[\bm{\lambda}(t)=\begin{cases}
\max [\arg\max  E(\cdot;t)-c(\cdot)], & \text{ if } \max  E(\cdot;t)-c(\cdot)\geq 0,\\
0, & \text{ otherwise.}\end{cases}\]

\begin{lem}\label{lem: without loss}
    For a transfer rule $t(\theta,a)$, there exists a transfer rule $\tilde t(\theta-a)$ that depends only on the difference $\theta-a$; given transfer rule $\tilde t$, the agent chooses  precision $\bm\lambda(t)$ when he is forced to reveal the signal. That is,
    $$\bm\lambda(t)=\hat{\bm\lambda}(\tilde t).$$
\end{lem}
\begin{proof}
Choose a measurable optimal reporting rule $\bm{a}(s;\lambda)$ after observing signal $s$ at precision $\lambda$:
$$\bm{a}(s;\lambda)\in \arg\max_a \int_{\mathbb R} t(\theta', a)\varphi(\theta';s,\lambda)d\theta'.$$
$$E(\lambda;t)=\int_{s\in A} \frac{1}{|A|}ds \int t(\theta', \bm a(s; \lambda))\varphi(\theta';s,\lambda)d\theta'$$
Since the reporting rule is measurable, the integrand is bounded, and $0<|A|<\infty$, this integral is well-defined. Let $\theta''= \theta'-s$,
\[
\begin{split}
    E(\lambda;t)
    &=\int_{s\in A} \frac{1}{|A|}ds \int t(\theta''+s, \bm a(s; \lambda))\varphi(\theta'';0,\lambda)d\theta''.
\end{split}
\]
Let $\tilde t(\theta'')= \int_{s\in A} \frac{1}{|A|} t(\theta''+s, \bm a(s; \bm\lambda(t))) ds$.
The transfer $\tilde t$ is measurable and takes values in $[0,1]$.
For each fixed $s$, the prediction error
$x+s-\bm a(s;\bm\lambda(t))$ tends to infinity in absolute value
as $|x|\to\infty$. The uniform vanishing assumption and dominated
convergence therefore imply that $\tilde t(x)\to0$.
By the Fubini--Tonelli theorem,
\[
\begin{split}
    E(\bm\lambda(t);t)&=\int_{s\in A} \frac{1}{|A|}ds \int t(\theta''+s, \bm a(s; \bm\lambda(t)))\varphi(\theta'';0,\bm\lambda(t))d\theta''\\
    &= \int_{\mathbb R} \tilde t(\theta'')\varphi(\theta'';0,\bm\lambda(t))d\theta''\\
    &=\hat E(\bm\lambda(t);\tilde t).
\end{split}
\]
Moreover, at $\lambda\neq \bm\lambda(t)$,
\[
\begin{split}
    E(\lambda;t)&=\int_{s\in A} \frac{1}{|A|}ds \int_{\mathbb R} t(\theta', \bm a(s; \lambda))\varphi(\theta';s,\lambda)d\theta'\\
    &\geq \int_{s\in A} \frac{1}{|A|}ds \int_{\mathbb R} t(\theta', \bm a(s; \bm\lambda(t)))\varphi(\theta';s,\lambda)d\theta'\\
    &=\int_{s\in A} \frac{1}{|A|}ds \int t(\theta''+s, \bm a(s; \bm\lambda(t)))\varphi(\theta'';0,\lambda)d\theta''\\
    &= \int_{\mathbb R} \tilde t(\theta'')\varphi(\theta'';0,\lambda)d\theta''\\
    &=\hat E(\lambda;\tilde t)
\end{split}
\]where the inequality follows by the definition of $\bm a(s;\lambda)$. By the definition of $\bm{\lambda}(t)$,
$\bm{\hat\lambda}(\tilde t)=\bm\lambda(t)$.
\end{proof}
Given this lemma, I can use the constructed $\tilde t$ in the proof of Theorem~\ref{thm: step}, where I show that there exists a cutoff transfer that induces a weakly higher precision.

The same construction also applies to the payment-minimization
problem in Section~\ref{sec: savings}. Fix a positive target
$\bar\lambda$ implemented by a transfer rule $t(\theta,a)$,
and replace $\bm\lambda(t)$ by $\bar\lambda$ in the construction
above. Denote the resulting transfer by $\tilde t$. Then
\[
    \hat E(\bar\lambda;\tilde t)=E(\bar\lambda;t),
    \qquad
    \hat E(\mu;\tilde t)\leq E(\mu;t)
    \quad\text{for all }\mu\geq0.
\]
Consequently, $\bar\lambda$ maximizes
$\hat E(\mu;\tilde t)-c(\mu)$, and the expected payment at the
target is unchanged. The transfer $\tilde t$ can therefore be
used in the marginal-incentive and payment-bound arguments in
the sufficiency proof of
Theorem~\ref{thm: cost minimizing cutoff}. Under increasing
elasticity below $1$, the cutoff constructed in that proof
implements the same target, induces truthful reporting, and
has an expected payment no greater than $E(\bar\lambda;t)$.
Thus, allowing transfers to depend separately on the state
and the report does not reduce the minimum implementation cost.

\vspace{5mm}
In the second approach, the state space is a circle in $\mathbb{R}^2$ with a circumference of $1$, and the prior distribution is uniform. This setup corresponds to cases where the state is periodic, such as the unit vector in $\mathbb R^2$ which corresponds to a direction in $\mathbb R^2$ (In Machina triangle, each direction corresponds to an indifference curve.), the time of day (e.g., 15:20) or the day of the year (e.g., October 19). The model remains almost the same. The only exception is that $\phi$ admits a compact support $[-M,M]$ and the precision $\lambda\in [2M,+\infty]$. The proof that it is without loss to assume that transfer depends only on the difference $\theta-a$ is almost the same as above. The exception is that the integral of $s$ in Lemma \ref{lem: without loss} can be taken over the entire circle.

\section{Omitted Proofs of Section~\ref{sec: prelim}-\ref{sec: extensions} }
\begin{proof}[Proof of Lemma \ref{lem: ratio}]
$$\frac{\partial \ln \phi(\lambda x)}{\partial\ln \lambda}=\frac{\phi'(\lambda x)}{\phi(\lambda x)}\lambda x=-\eta (\lambda x).$$
Thus, we have 
$$\frac{\partial [\ln\frac{\phi(\lambda x_1)}{\phi(\lambda x_2)}]}{\partial\ln \lambda}=\eta (\lambda x_2)-\eta (\lambda x_1).$$
The second part follows by
$$\frac{\partial [\frac{\varphi(x_1;0,\lambda)}{\varphi(x_2;0,\lambda)}]}{\partial \lambda}=\frac{\partial [\frac{\phi(\lambda x_1)}{\phi(\lambda x_2)}]}{\partial \lambda}.$$
\end{proof}

\begin{proof}[Proof of Lemma \ref{lem: cross derivative}]
Let $\Phi$ denote the anti-derivative of $\phi$, i.e., $\Phi(x)= \int_{-\infty}^x \phi(y)dy$. We compute the expected transfer given the cutoff transfer $d$ 
$$E(\lambda; d)=2\Phi(\lambda d)-1.$$
The derivative of the expected transfer with respect to $\lambda$ is
$$\frac{\partial E(\lambda; d)}{\partial \lambda}=2d \phi(\lambda d).$$
Let us see how the derivative changes when varying $d$,
$$\frac{\partial \frac{\partial E(\lambda; d)}{\partial \lambda}}{\partial d}=2\phi(\lambda d)[1-\eta (\lambda d)].$$
\end{proof}

\begin{lem}\label{lem: expected transfer}
    The agent's expected transfer $$
    E(\lambda;t) = \max_{a \in \mathbb{R}} \int_{\mathbb R} t(\theta'- a) \varphi(\theta'; s,\lambda) d\theta'
$$  is well-defined.
\end{lem}
\begin{proof}[Proof of Lemma \ref{lem: expected transfer}]
    I show that the maximum is well-defined, given that the transfer rule $t$ vanishes at infinity. Fix a signal $s$ and precision $\lambda$. Let $T(a)$ be the agent's expected transfer after reporting $a$, $$T(a)= \int t(\theta'-a )\varphi (\theta';s,\lambda)d\theta'.$$
If $T(a)=0$ for all $a\in\mathbb R$, the $\arg\max_a T(a)$ is well-defined. Suppose that there exists $ a_1\in \mathbb R$ such that $T(a_1)>0$. As $t$ vanished at infinity, there exists a $K_1>0$ such that for all $|x|>K_1$, $t(x)< \frac{T(a_1)}{2}$. Moreover, there exists a $K_2>0$ such that $\Phi(\lambda K_2)> 1-\frac{T(a_1)}{4}$. Let $K=\max(K_1,K_2)$. If $|a-s|>2K$,
\[
\begin{split}
    T(a)=& \int t(\theta'-a )\varphi (\theta';s,\lambda)d\theta'\\
    =&  \int_{\theta':|\theta'-s|<K} t(\theta'-a )\varphi (\theta';s,\lambda)d\theta'+\int_{\theta':|\theta'-s|\geq K} t(\theta'-a )\varphi (\theta';s,\lambda)d\theta'\\
    \leq & \frac{T(a_1)}{2} \int_{\theta':|\theta'-s|<K} \varphi (\theta';s,\lambda)d\theta'+1\cdot \int_{\theta':|\theta'-s|\geq K} \varphi (\theta';s,\lambda)d\theta'\\
    \leq  & \frac{T(a_1)}{2} +2[1- \Phi(\lambda K)]\\
    =& T(a_1).
\end{split}
\]
Thus, $$\arg\max_{a\in\mathbb R} T(a)=\arg\max_{|a-s|\leq 2K} T(a).$$
Next, I show that $T(\cdot)$ is continuous. 
\[
\begin{split}
    |T(a+\epsilon)-T(a)|=&\bigg|\int t(\theta'-a-\epsilon )\varphi (\theta';s,\lambda)d\theta'-\int t(\theta'-a )\varphi (\theta';s,\lambda)d\theta'\bigg|\\
    =&\bigg|\int t(\theta'-a )\varphi (\theta'+\epsilon;s,\lambda)d\theta'-\int t(\theta'-a )\varphi (\theta';s,\lambda)d\theta'\bigg|\\
    =&\bigg|\int t(\theta'-a )\varphi (\theta';s-\epsilon,\lambda)d\theta'-\int t(\theta'-a )\varphi (\theta';s,\lambda)d\theta'\bigg|\\
    =& \bigg|\int t(\theta'-a )[\varphi (\theta';s-\epsilon,\lambda)-\varphi (\theta';s,\lambda)]d\theta'\bigg|\\
    \leq & \int \bigg|\varphi (\theta';s-\epsilon,\lambda)-\varphi (\theta';s,\lambda)\bigg|d\theta'\\
    = & 2 \int_{\theta'\leq s-\frac{\epsilon}{2}} \varphi (\theta';s-\epsilon,\lambda)-\varphi (\theta';s,\lambda)d\theta'\\
    =& 2[\Phi(\frac{\epsilon}{2}\lambda)-\Phi(-\frac{\epsilon}{2}\lambda)]
\end{split}
\]
Thus, $T$ is continuous as the CDF $\Phi$ is continuous. A continuous function achieves a maximum on a compact set. So the maximum in $E(\lambda;t)$ is well-defined.
\end{proof}

Let $\hat E(\lambda;t)$ denote the agent's expected transfer when he has to report truthfully
$$\hat E(\lambda; t)=\int_{\mathbb R}  t(\theta') \varphi (\theta'; 0, \lambda)d\theta'.$$
Let $\bm{\hat\lambda}(t)$ denote the corresponding induced precision
\[\bm{\hat\lambda}(t)=\begin{cases}
\max [\arg\max \hat E(\cdot;t)-c(\cdot)], & \text{ if } \max \hat E(\cdot;t)-c(\cdot)\geq 0,\\
0, & \text{ otherwise.}\end{cases}\]

\begin{lem}\label{lem: aux fix}
    Suppose $\phi$ satisfies increasing elasticity above 1. Suppose the agent has to report truthfully. Given a symmetric transfer rule $\bar t$ and induced $\lambda^*=\bm{\hat\lambda}(\bar t)$, the new transfer
    \[
    \bar t'(x)=\begin{cases}
1, & \text{ if } \lambda^* |x|< \eta ^{-1}(1),\\
\bar t(x), & \text{ otherwise }   \end{cases}\]
can induce a weakly higher precision $\bm{\hat\lambda}(\bar t')\geq \bm{\hat\lambda}(\bar t)$.
    
\end{lem}

\begin{proof}[Proof of Lemma \ref{lem: aux fix}]
Let $\Delta t= \bar t'-\bar t$. I shall show that 
$$\hat E(\lambda^*;\Delta t)\geq \hat E(\lambda;\Delta t)\quad \text{for all } \lambda\leq \lambda^*.$$
This implies that $\bm{\hat\lambda}(\bar t')\geq \bm{\hat\lambda}(\bar t)$ since $\bm{\hat\lambda}(t)=
\max [\arg\max \hat E(\cdot;t)-c(\cdot)]$.

Consider two cutoff transfers $0<d_1<d_2\leq \eta ^{-1}(1)/\lambda^*$. For all $\lambda\leq \lambda^*$, we have $\eta (\lambda d_1)\leq 1$, $\eta (\lambda d_2)\leq 1$. By Lemma \ref{lem: cross derivative},
$$\frac{\partial^2 E(\lambda; d)}{\partial \lambda\partial d}\geq 0\quad \Leftrightarrow \quad \eta (\lambda d)\leq 1.$$
This implies 
$\frac{\partial E(\lambda; d)}{\partial \lambda}$ increases in $d\in[d_1,d_2]$ for all $\lambda\leq \lambda^*$, which implies $E(\lambda;d_2)-E(\lambda;d_1)$ increases in $\lambda\in[0,\lambda^*]$. As $\Delta t$ is an integral of such difference of cutoff transfers and $\hat E(\lambda;t)$ is linear in $t$, $\hat E(\lambda;\Delta t)$ increases in $\lambda\in[0,\lambda^*]$.

\end{proof}

\begin{lem}\label{lem: ratio change}
Suppose $\phi$ satisfies increasing elasticity above 1. Let $\eta ^{-1}(1)<x_1<x_2$. Then
$$\frac{\phi(\lambda' x_2)}{\phi(\lambda' x_1)}\geq \frac{\phi(x_2)}{\phi(x_1)}\quad \text{for all} \quad \lambda'<1.$$
\end{lem}

\begin{proof}[Proof of Lemma \ref{lem: ratio change}]
By Lemma \ref{lem: ratio}, 
$$\ln \frac{\phi(\lambda' x)}{\phi(x)}=\int_{\lambda'}^1 \eta (\lambda x)d\ln\lambda.$$
Then,
$$\frac{\phi(\lambda' x_1)}{\phi(x_1)}\leq \frac{\phi(\lambda' x_2)}{\phi(x_2)}\quad \text{for all} \quad \lambda'<1$$ is equivalent to
$$\int_{\lambda'} ^1 \eta (\lambda x_1)d\ln\lambda\leq \int_{\lambda'} ^1 \eta (\lambda x_2)d\ln\lambda\quad \text{for all} \quad \lambda'<1.$$
Let $y= \ln\lambda$, then it suffices to show
$$\int_{\ln\lambda'} ^0 \eta (x_1\exp(y) )dy\leq \int_{\ln\lambda'} ^0 \eta (x_2\exp(y))dy\quad \text{for all} \quad \lambda'<1.$$
Let $\Delta y= \ln(x_2/x_1)>0$. Then it suffices to show
$$\int_{\ln\lambda'} ^0 \eta (x_2\exp(y-\Delta y) )dy\leq \int_{\ln\lambda'} ^0 \eta (x_2\exp(y))dy\quad \text{for all} \quad \lambda'<1.$$
Let $g(y)= \eta (x_2 \exp(y))$. Then it suffices to show
\begin{equation}\label{eq: g function}
    \int_{\ln\lambda'} ^0 g(y-\Delta y) dy\leq \int_{\ln\lambda'} ^0 g(y)dy\quad \text{for all} \quad \lambda'<1.
\end{equation}
Note that function $g(y)$ is increasing in $y\in[\ln\frac{\eta ^{-1}(1)}{x_2},0]$ and single-crosses $1$ from below at $\ln\frac{\eta ^{-1}(1)}{x_2}$.

If $\ln\lambda'\geq -\Delta y$, then for all $y\in [\ln\lambda',0]$,
$$g(y)\geq g(\ln\lambda')\geq g(-\Delta y)\geq g(y-\Delta y).$$

Suppose $\ln\lambda'<- \Delta y$. Equation \eqref{eq: g function} is
$$\int_{\ln\lambda'-\Delta y}^{-\Delta y} g(y) dy\leq \int_{\ln\lambda'} ^0 g(y)dy\quad \text{for all} \quad \lambda'<1.$$
I can subtract $\int_{\ln\lambda'}^{-\Delta y} g(y) dy$ from both sides. Then the inequality becomes
$$\int_{\ln\lambda'-\Delta y}^{\ln\lambda'} g(y) dy\leq \int_{-\Delta y}^0 g(y)dy\quad \text{for all} \quad \lambda'<1.$$ It holds since for all $y\in[-\Delta y,0]$, $y'\in [\ln\lambda'-\Delta y,\ln\lambda']$,
$$g(y)\geq g(-\Delta y)\geq g(\ln\lambda')\geq g(y').$$

\end{proof}

\begin{proof}[Proof of Theorem \ref{thm: step} (3) implies (1) and Proof of Theorem \ref{thm: optimal step function}]
Suppose that $\phi$ satisfies increasing elasticity above $1$. If there is no $t$ such that $\bm\lambda(t)> 0$, then the problem is trivial.\footnote{In this case, we adopt the convention that every transfer is optimal.} Now suppose there exists a $t$ such that $\bm\lambda(t)> 0$. In step 1, I shall show that for all transfer $t$ with $\bm\lambda(t)> 0$, I can construct a new cutoff transfer $d$ such that $\bm{\lambda}(d)\geq \bm{\lambda}(t)$. That is, cutoff transfer $d$ weakly improves over $t$. In step 2, I shall optimize over cutoff transfers and prove Theorem \ref{thm: optimal step function}. In step 3, I show uniqueness under stronger conditions.

\textbf{Step 1: Optimality of Cutoff Transfers}

I first provide a roadmap for this proof. First, I construct an auxiliary transfer $\tilde t$ that preserves the induced precision when the agent has to report truthfully. Second, I symmetrize $\tilde t$ to obtain $\bar t$. Third, I fill the low prediction-error region up to $\frac{\eta^{-1}(1)}{\bm\lambda(t)}$ to obtain $\bar t'$. Fourth, I replace $\bar t'$ by a cutoff with the same payoff at $\bm\lambda(t)$. Finally, I use Lemma \ref{lem: ratio change} to compare incentives below $\bm\lambda(t)$.

Fix a transfer rule $t$.  First, I show that there exists a transfer rule $\tilde t(\theta-a)$ such that given $\tilde t$, the agent chooses precision $\bm\lambda(t)$ when he has to report truthfully. That is, $$\hat{\bm\lambda}(\tilde t)=\bm\lambda(t).$$
For the transfer rule $t$ that depends on the state and the report, Lemma \ref{lem: without loss} gives us the desired $\tilde t$. Now assume that $t$ depends only on the difference between the state and the report, $\theta-a$. The construction of $\tilde t$ is simpler and is as follows. Let $\bm{a}(s;\bm{\lambda}(t))$ denote the agent's report after observing signal $s$ at the precision $\bm{\lambda}(t)$. 
$$\bm{a}(s;\bm{\lambda}(t))= \arg\max_a \int_{\mathbb R} t(\theta'- a)\varphi(\theta';s,\bm{\lambda}(t))d\theta'.$$
As both $t$ and $\varphi$ are translation invariant, $\bm{a}(s;\bm{\lambda}(t))-s$ is a constant. I define a new transfer rule $\tilde t$ by
$$\tilde t(x)= t(x+s-\bm{a}(s;\bm{\lambda}(t))).$$
Given transfer $\tilde t$, the agent reports truthfully at $\bm{\lambda}(t)$. Note that replacing $t$ by $\tilde t$ only changes the agent's report by a constant. Two transfers $t$ and $\tilde t$ provide the same incentive for the agent to choose $\lambda$, i.e., $$E(\cdot; t)=E(\cdot; \tilde t),\quad\bm{\lambda}(\tilde t)=\bm{\lambda}(t).$$ 
I define an auxiliary problem: the agent has to reveal his signal. The agent's expected transfer is $$\hat E(\lambda;\tilde t)=\int_{\mathbb R}  \tilde t(\theta') \varphi (\theta'; 0, \lambda)d\theta'.$$
It must be true that for all $\lambda \in \mathbb R_+$,
\begin{equation}\label{eq: auxiliary inequality}
    \hat E(\lambda; \tilde t)\leq E(\lambda;  \tilde t)
\end{equation}
and 
\begin{equation}\label{eq: auxiliary equality}
    \hat E(\bm{\lambda}(t); \tilde t)= E(\bm{\lambda}(t);  \tilde t).
\end{equation}
The inequality holds as the agent loses the flexibility to misreport in the auxiliary problem. The equality holds as the agent reports truthfully at $\lambda=\bm{\lambda}(t)$. By Equation \eqref{eq: auxiliary inequality} and \eqref{eq: auxiliary equality}, $$\bm{\hat \lambda}(\tilde t)=\bm{\lambda}(\tilde t)=\bm{\lambda}(t).$$

Then, I symmetrify the transfer rule $\tilde t$ to obtain
$$\bar t(x)= \frac{1}{2}[\tilde t(x)+\tilde t(-x)].$$
For all $\lambda \in \mathbb R_+$, we have
$$\hat E(\lambda;\bar t)=\hat E(\lambda;\tilde t).$$
$$\bm{\hat \lambda}(\bar t)=\bm{\hat \lambda}(\tilde t)=\bm{\lambda}(t).$$
By Lemma \ref{lem: aux fix}, I can augment transfer $\bar t$ to a new transfer $\bar t'$
\[
    \bar t'(x)=\begin{cases}
1, & \text{ if } \bm{\lambda}(t) |x|< \eta ^{-1}(1),\\
\bar t(x), & \text{ otherwise }   \end{cases}\]
and induces higher precision $$\bm{\hat \lambda}(\bar t')\geq \bm{\hat \lambda}(\bar t)=\bm{\lambda}(t).$$

Then I can construct a cutoff transfer $d$. I pin down the cutoff by requiring that cutoff transfer $d$ and $\bar t'$ offers the same expected transfer at $\bm{\lambda}(t)$, i.e.,
\begin{equation}\label{eq: construct step}
    E(\bm{\lambda}(t);d)= \hat E(\bm{\lambda}(t); \bar t').
\end{equation}
Let $\Delta t(x)= 1_{|x|\leq d}-\bar t'(x)$ denote the transfer difference. The transfer difference must take the form
\[\Delta t(x)\begin{cases}
= 0, & \text{ if } |x|\leq \eta ^{-1}(1)/\bm\lambda(t),\\
\geq 0, & \text{ if } \eta ^{-1}(1)/\bm\lambda(t)<|x|\leq d,\\
\leq 0, & \text{ otherwise.}\end{cases}\]
Compare the expected transfer given cutoff transfer $d$ and the expected transfer $\bar t'$ in the auxiliary problem.
\[
\begin{split}
    E(\lambda;d)-\hat E(\lambda;\bar t')=\int_{\mathbb R} \Delta t(\theta')\varphi (\theta'; 0,\lambda)d\theta'.
\end{split}
\]
Consider $y_1,\,y_2$ such that $\eta ^{-1}(1)/\bm\lambda(t)<y_1<y_2$. Let $x_1= \bm\lambda(t) y_1$, $x_2= \bm\lambda(t) y_2$. Then $\eta ^{-1}(1)<x_1<x_2$. By Lemma \ref{lem: ratio change}, for all $\lambda<\bm\lambda(t)$, we have
$$\frac{\phi( x_2 \lambda/\bm\lambda(t))}{\phi( x_1 \lambda/\bm\lambda(t))}\geq \frac{\phi(x_2)}{\phi(x_1)}.$$
$$\frac{\phi(\lambda y_2)}{\phi(\lambda y_1)}\geq \frac{\phi(\bm\lambda(t) y_2)}{\phi(\bm\lambda(t) y_1)}.$$
$$\frac{\varphi(y_2;0, \lambda )}{\varphi(y_1;0,\lambda )}\geq \frac{\varphi( y_2;0,\bm\lambda(t))}{\varphi( y_1;0,\bm\lambda(t))}.$$
This implies
$$E(\lambda;d)-\hat E(\lambda;\bar t')\leq 0,  \text{ if } \lambda< \bm{\lambda}(t).$$
Given $\bm{\hat \lambda}(\bar t')\geq \bm{\lambda}(t)$
and the definition $\bm{\lambda}(t)=
\max [\arg\max E(\cdot;t)-c(\cdot)]$, we have $\bm{\lambda} (d)\geq \bm{\lambda} (t)$ (see Figure \ref{pic expected transfer}).

\begin{figure}[htp]
    \centering
    \begin{tikzpicture}[scale=1]
        \draw[thick, ->] (0,0) -- (0,6) node[left]{$E$};
        \draw[thick, ->] (0,0) -- (9,0) node[right]{ $\lambda$};
        \draw (7, 3) .. controls (5,1) .. (1,0.3) ;
        \node[below] at (7.5,3){$c(\lambda)$};
        \draw[color=blue] (1,0.4) .. controls (4,3.5) .. (7,4)node[right]{\footnotesize $E(\lambda; t)$} ;
        \draw[color=red] (1.5,0.1) .. controls (4,3.62) .. (7,3.6)node[right]{\footnotesize $\hat E(\lambda;  \tilde t)=\hat E(\lambda; \bar t)$} ;
        \node[below,color=red] at (4.2,0) {$\bm{\lambda}(t)$};
        \draw[dotted] (4.2,0) -- (4.2,4.1);
        \fill (4.2,3.26) circle (0.05);
        \draw[color=brown] (1.5,0.3) .. controls (4,4.5) .. (7,5)node[right]{\footnotesize $\hat E(\lambda; \bar t')$} ;
        \node[below,color=brown] at (5.3,0) {$\bm{\hat \lambda}(\bar t')$};
        \draw[dotted] (5.3,0) -- (5.3,4.7);
        \fill (5.3,4.67) circle (0.05);  
        \draw (2,0.2) .. controls (4.5,5) .. (7,6)node[right]{\footnotesize $ E(\lambda; d)$} ;
        \fill (4.2,4.1) circle (0.05);  
        
    \end{tikzpicture}
    \caption{\small Expected Transfer} \label{pic expected transfer}
\end{figure}
\vspace{5mm}
\textbf{Step 2: Proof of Theorem \ref{thm: optimal step function}}

Next, I optimize over cutoff transfers to prove Theorem \ref{thm: optimal step function} and show the existence of the optimal cutoff transfer. If we set $d<\bar d$, the agent never chooses to work. So consider $d\geq \bar d$. 

\textbf{Case 1: $\bm\lambda(\bar d) \bar d\geq \eta ^{-1}(1).$}

Let $\tilde \lambda = \eta ^{-1}(1)/\bar d\leq \bm\lambda (\bar d)$. For any $d>\bar d$, if $\bm\lambda(d)\leq \tilde\lambda$, then $\bm\lambda(d)\leq \bm\lambda (\bar d)$. Consider $\bm\lambda(d)> \tilde\lambda$. For all $d\geq \bar d$ and $\lambda\geq \tilde\lambda$,
$$\eta (\lambda d)\geq \eta (\bar d\tilde \lambda)\geq 1 \quad \Rightarrow \quad\frac{\partial^2 E(\lambda; d)}{\partial \lambda\partial d}\leq 0.$$
Since $\bm{\lambda}(d)=
\max [\arg\max E(\cdot;d)-c(\cdot)]$. By the Topkis' monotone comparative statics theorem, $\bm\lambda(d)\leq \bm\lambda(\bar d)$ for all $d\geq \bar d$. Thus, $\bar d$ is the optimal cutoff.

\textbf{Case 2: $ \bm\lambda(\bar d) \bar d< \eta^{-1}(1)$.} For all $\bar d\leq d\leq d^*$ and $\lambda\leq \frac{\eta ^{-1}(1)}{d}$,
$$\eta (\lambda d)\leq 1\quad \Rightarrow \quad\frac{\partial^2 E(\lambda; d)}{\partial \lambda\partial d}\geq 0 .$$
This implies $\bm\lambda(d)$ is increasing in $d$ for $\bar d\leq d\leq d^*$. Now suppose $d> d^*$. If $\bm\lambda(d) d \leq \eta ^{-1}(1)$, then $$\bm\lambda(d)\leq \eta ^{-1}(1)/d\leq \eta ^{-1}(1)/d^*=\bm\lambda(d^*).$$ If $\bm\lambda(d) d > \eta ^{-1}(1)$, for all $d\geq d^*$ and $ \lambda\geq \eta ^{-1}(1)/d$,
$$\frac{\partial^2 E(\lambda; d)}{\partial \lambda\partial d}\leq 0\quad \Rightarrow \quad \bm\lambda(d)\leq \bm\lambda(d^*) $$
by the Topkis' monotone comparative statics theorem.

Now, I prove existence. In case 1, $\bar d$ is the optimal cutoff transfer. In case 2, I can increase $d$ starting from $\bar d$. Before hitting $d^*$, $\bm\lambda(d)$ increases in $d$. As $\eta^{-1}(1)$ is well-defined and finite given that $\phi$ satisfies increasing elasticity above $1$. Eventually, increasing $\bm\lambda(d) d$ can hit $\eta^{-1}(1)$. Thus, $d^*$ is well defined.

\vspace{5mm}
\textbf{Step 3: Uniqueness}

Finally, suppose that $\phi$ satisfies strictly increasing elasticity above $1$, the cost function is continuously differentiable, and the optimal precision is an interior solution (not $0$ or $+\infty$).\footnote{Distribution $\phi$ satisfies \textit{strictly increasing elasticity above 1} if $\eta (\cdot)$ single-crosses $1$ from below, $\{x|\eta(x)=1\}$ is a singleton, and is strictly increasing after the cross.} I show that the optimal transfer rule is unique. First, I show that for any transfer rule $t$ that is not a cutoff transfer, the cutoff transfer constructed above induces a strictly larger precision: $\bm{\lambda}(d)>\bm{\lambda}(t)$. Since $t$ is not a cutoff transfer, at least one of $\bar t'-\bar t$ and $\Delta t$ is a non-zero function. If $\bar t'-\bar t$ is non-zero,
\begin{equation}\label{eq: unique 1}
    \frac{\partial [\hat E(\lambda;\bar t')-\hat E(\lambda;\bar t)]}{\partial \lambda}\bigg|_{\lambda=\bm\lambda(t)}>0.
\end{equation}
This follows by a strict version of Lemma \ref{lem: aux fix} in which strictly increasing elasticity above $1$ implies
$$\eta (\lambda d)< 1\quad \Rightarrow \quad\frac{\partial^2 E(\lambda; d)}{\partial \lambda\partial d}> 0 .$$
As $\hat E(\cdot;\bar t)-c(\cdot)$ is continuously differentiable,\footnote{$\hat E(\cdot;\bar t)$ is continuously differentiable as $\phi$ is continuously differentiable.} $\hat{\bm\lambda}(\bar t)$ is interior, Equation \eqref{eq: unique 1} implies $\hat{\bm\lambda}(\bar t')>\hat{\bm\lambda}(\bar t)$ by \cite{edlin1998strict}. This implies $\bm\lambda(d)>\bm\lambda(t)$.

If $\Delta t$ is non-zero, next I show  
\begin{equation}\label{eq: unique 2}
    \frac{\partial [E(\lambda;d)-\hat E(\lambda;\bar t')]}{\partial \lambda}\bigg|_{\lambda=\bm\lambda(t)}>0.
\end{equation}
By the construction of $\Delta t$, 
\begin{equation}\label{eq: construction Delta t}
    \int \Delta t(x) \varphi(x;0,\bm\lambda(t)) dx=0
\end{equation}
$$\int \Delta t(x)  \phi(\bm\lambda(t) x) dx=0.$$
Take some small $\epsilon>0$. Let $\lambda'=\bm\lambda(t) \exp(\epsilon)$. Consider 
$$E(\lambda';d)-\hat E(\lambda';\bar t')=\int \Delta t (x)\lambda' \phi(\lambda' x)dx.$$ By $\ln(\lambda')-\ln(\bm\lambda(t))=\epsilon$ and Lemma \ref{lem: ratio}, 
$$\ln\phi(\lambda' x)-\ln\phi(\bm\lambda(t)x)=-\eta(\bm\lambda(t) |x|))\epsilon+\smallO(\epsilon).$$
Thus,
\[
\begin{split}
    E(\lambda';d)-\hat E(\lambda';\bar t')&=\int \Delta t (x)\lambda' \phi(\lambda' x)dx\\
    &=\lambda' \int \Delta t (x) \phi(\bm\lambda(t) x)\exp(-\epsilon\eta(\bm\lambda(t) |x|))dx+\smallO(\epsilon)\\
    &=-\lambda' \int \Delta t (x) \phi(\bm\lambda(t) x)\epsilon\eta(\bm\lambda(t) |x|)dx+\smallO(\epsilon),
\end{split}
\]
where the last equality follows by Equation \eqref{eq: construction Delta t}.
\[
\begin{split}
    \lim_{\epsilon\to 0}\frac{E(\lambda';d)-\hat E(\lambda';\bar t')}{\epsilon}
    &=-\bm\lambda(t) \int \Delta t (x) \phi(\bm\lambda(t) x)\eta(\bm\lambda(t) |x|)dx\\
    &=- \int \Delta t (x) \varphi(x;0,\bm\lambda(t) )\eta(\bm\lambda(t) |x|)dx
\end{split}
\]
$$\frac{\partial [E(\lambda;d)-\hat E(\lambda;\bar t')]}{\partial \lambda}\bigg|_{\lambda=\bm\lambda(t)}=-\frac{1}{\bm\lambda(t)} \int \Delta t (x) \varphi(x;0,\bm\lambda(t) )\eta(\bm\lambda(t) |x|)dx$$
which is strictly positive by Equation \eqref{eq: construction Delta t}, $\eta(y)$ being strictly increasing for $y>\eta^{-1}(1)$, and $\Delta t(x)$ being supported on $|x|\geq \frac{\eta^{-1}(1)}{\bm\lambda(t)}$. Again by $\hat E(\cdot;\bar t')-c(\cdot)$ being continuously differentiable, $\hat{\bm\lambda}(\bar t')$ is interior, Equation \eqref{eq: unique 2} implies $\bm\lambda(d)>\hat{\bm\lambda}(\bar t')\geq \bm\lambda(t)$ by \cite{edlin1998strict}.

I have shown that for any transfer rule $t$ that is not a cutoff transfer, $\bm{\lambda}(d)>\bm{\lambda}(t)$. Thus, any optimal transfer must be a cutoff transfer. Given strictly increasing elasticity above $1$, cost function being continuously differentiable, optimal precision being interior, all comparative statics analyses in the proof of Theorem \ref{thm: optimal step function}  are strict. Thus, there exists a unique optimal transfer rule, which is a cutoff transfer.

\end{proof}

Next, I prove a lemma which shall be critical for proving the necessity of Theorem \ref{thm: step}, statement (2) implies (3).

\begin{lem}\label{lem: monotone}
Suppose that for every increasing, convex, and continuously differentiable cost function, there exists an optimal transfer that is a cutoff transfer. For an increasing and convex cost function $c\in C^1$, let $d^*$ be the optimal cutoff and $\lambda^*$ be the induced precision with $\eta(\lambda^* d^*)\geq 1$. Then for all $x_1\in [0,\lambda^* d^*)$ and $x_2\in [\lambda^* d^*,+\infty)$,
$$\eta (x_1)\leq \eta (x_2).$$

\end{lem}

\begin{proof}[Proof of Lemma \ref{lem: monotone}]
Without loss I can assume $c(\lambda^*)=E(\lambda^*;d^*)$ since I can increase $c$ by a constant without affecting $d^*$ and $\lambda^*$. Similarly, without loss of generality, I assume that $\lambda^*$ uniquely maximizes $E(\lambda;d^*)-c(\lambda)$.\footnote{Otherwise, I can always increase $c(\lambda)$ for $\lambda\neq\lambda^*$, without affecting $d^*$ and $\lambda^*$.}

Towards a contradiction, suppose $\exists\, x_1\in [0,\lambda^* d^*),\,x_2\in [\lambda^* d^*,+\infty)$ such that 
$\eta (x_1)> \eta (x_2)$. As $\eta (\cdot)$ is continuous, I can pick $\delta_1>0$, $\delta_2>0$ small enough such that for all $y_1\in [\frac{x_1}{\lambda^*}-\delta_1, \frac{x_1}{\lambda^*}+\delta_1]$ and $y_2\in [\frac{x_2}{\lambda^*}-\delta_2, \frac{x_2}{\lambda^*}+\delta_2]$,
\begin{equation}\label{eq: order}
    \eta (y_1\lambda^*)> \eta (y_2\lambda^*)
\end{equation}
and
$$\int_{\frac{x_1}{\lambda^*}-\delta_1}^{\frac{x_1}{\lambda^*}+\delta_1} \varphi (x;0,\lambda^*) dx=\int_{\frac{x_2}{\lambda^*}-\delta_2}^{\frac{x_2}{\lambda^*}+\delta_2} \varphi (x;0,\lambda^*) dx.$$
I define $$\Delta t(x)= -1_{|x|\in [\frac{x_1}{\lambda^*}-\delta_1, \frac{x_1}{\lambda^*}+\delta_1]}+1_{|x|\in [\frac{x_2}{\lambda^*}-\delta_2, \frac{x_2}{\lambda^*}+\delta_2]}$$
and a new transfer $\underline t(x)=1_{|x|\leq d^*}+\Delta t(x)$. 

I can set $\delta_1$ and $\delta_2$ to be small enough such that the agent reports truthfully given the transfer $\underline t$, since $\varphi'(d^*;0,\lambda^*)<0$ is bounded from above due to $\eta (\lambda^* d^*)\geq 1$. To see this, it suffices to consider the case $\lambda^*=1$. If the agent reports truthfully, the expected transfer is $2\Phi(d^*)-1$. Now suppose the agent misreports by $\epsilon'>0$. Then his expected transfer is
\[
\begin{split}
    &\Phi(d^*-\epsilon')+\Phi(d^*+\epsilon')-1\\
    +&\Phi(x_2-\epsilon'+\delta_2)-\Phi(x_2-\epsilon'-\delta_2)+\Phi(x_2+\epsilon'+\delta_2)-\Phi(x_2+\epsilon'-\delta_2)\\
    -&[\Phi(x_1-\epsilon'+\delta_1)-\Phi(x_1-\epsilon'-\delta_1)]-[\Phi(x_1+\epsilon'+\delta_1)-\Phi(x_1+\epsilon'-\delta_1)],
\end{split}
\]
where the second and the third line can be made arbitrarily close to $0$ as $\delta_1$ and $\delta_2$ tend to $0$, since $\Phi$ is continuous. Moreover, the difference 
$$\Phi(d^*-\epsilon')+\Phi(d^*+\epsilon')-2\Phi(d^*)$$
is strictly negative and is strictly decreasing in $\epsilon'$, due to $\phi$ being single-peaked and $\eta(d^*)=-\frac{\phi'(d^*)}{\phi(d^*)}d^*\geq 1$. Therefore, the expected transfer under misreport is lower than $2\Phi(d^*)-1$ when $\delta_1$ and $\delta_2$ are small enough.

Then we have
$$ E(\lambda^*;\underline t)=E(\lambda^*;d^*).$$
By Lemma \ref{lem: ratio} and Equation \eqref{eq: order},
$$\frac{\partial [E(\lambda;\underline t)-E(\lambda;d^*)]}{\partial \lambda}\bigg|_{\lambda=\lambda^*}>0.$$ 
Since previously $\lambda^*$ uniquely maximizes $E(\lambda;d^*)-c(\lambda)$, we can pick $\delta_1$ and $\delta_2$ small enough such that $\bm\lambda(\underline t)\in (\lambda^*-\epsilon, \lambda^*+\epsilon)$, for some small $\epsilon>0$. As $c$ is continuously differentiable, $E(\lambda^*;\underline t)=E(\lambda^*;d^*)$, $$\frac{\partial [E(\lambda;\underline t)-E(\lambda;d^*)]}{\partial \lambda}\bigg|_{\lambda=\lambda^*}>0,$$ by \cite{edlin1998strict} we have $\bm\lambda(\underline t)> \lambda^*$,
contradicting that $(d^*,\lambda^*)$ is optimal (Figure \ref{pic monotone}).

\begin{figure}[htp]
    \centering
    \begin{tikzpicture}[scale=0.8]
        \draw[thick, ->] (0,0) -- (0,6) node[left]{$E$};
        \draw[thick, ->] (0,0) -- (9,0) node[right]{ $\lambda$};
        \draw (6, 6)node[right]{$c(\lambda)$} .. controls (5,3) .. (1,2.3) ;
        \draw[color=blue] (5.7,4.4)node[right]{\footnotesize $E(\lambda; \underline t)$} .. controls (4.3,3.2)..(2.2,0.1) ;
        \draw[color=red] (6,3.6)node[right]{\footnotesize $E(\lambda; d^*)$} .. controls (4.2,3.25)..(2,0.8) ;
        \fill (4.4,3.15) circle (0.05);  
        \node[below] at (4.4,0) {$\lambda^*$};
        \draw[dotted] (4.4,0) -- (4.4,4);
    \end{tikzpicture}
    \caption{\small Expected Transfer} \label{pic monotone}
\end{figure}
\end{proof}

\begin{proof}[Proof of Theorem \ref{thm: step} (2) implies (3)]

Suppose that for every increasing, convex, and continuously differentiable cost function, there exists an optimal transfer that is a cutoff transfer. Take a increasing, convex cost function $c_0\in C^1$, with $d_0$ and $\lambda_0$ being the corresponding optimal cutoff and induced precision. I can pick $c_0$ such that $\lambda_0$ is the unique maximizer of\footnote{Given any increasing and convex cost function $c_0\in C^1$, I can always construct a new increasing, convex, and continuously differentiable cost function by increasing $c_0(\lambda)$ for all $\lambda\neq \lambda_0$ while keeping $c_0(\lambda_0)$ and $\frac{dc_0}{d\lambda}|_{\lambda=\lambda_0}$ unchanged.} $E(\cdot;d_0)-c_0(\cdot)$. As $E(\cdot;d_0)-c(\cdot)$ is continuously differentiable, I can use the first-order approach. As 
$$\frac{\partial^2 E(\lambda; d)}{\partial \lambda\partial d}\leq 0\quad \Leftrightarrow \quad \eta (\lambda d)\geq 1,$$
we have $\eta (\lambda_0 d_0)\geq 1$. If not, I can slightly increase $d_0$, which leads to a larger $$\frac{\partial E(\lambda; d)}{\partial \lambda}\bigg|_{\lambda=\lambda_0}.$$
This induces a larger $\lambda>\lambda_0$, a contradiction.\footnote{As $\lambda_0$ is the unique maximizer of $E(\cdot;d_0)-c_0(\cdot)$, I can always make the increment on $d_0$ small enough such that the new maximizer never lies in $[0,\lambda_0)$.}

Since the cutoff transfer $d_0$ is the optimal transfer and $c_0\in C^1$, by Lemma \ref{lem: monotone} and $\eta (\cdot)$ being continuous,
$$\eta (x)\leq \eta (\lambda_0 d_0)\quad \text{if } x\leq \lambda_0 d_0$$
$$\eta (x)\geq \eta (\lambda_0 d_0)\quad \text{if } x\geq \lambda_0 d_0.$$
Note that $\eta (0^+)<1$ (as $\phi$ is integrable around $0$), $\eta (\lambda_0 d_0)\geq 1$, $\eta (\cdot)$ is continuous, I can find the largest $x$ where $\eta (\cdot)$ crosses $1$ from below
$$x_1= \min\{x|\eta (x')\geq 1, \text{ if } x'\geq x\}\leq \lambda_0 d_0.$$
$\eta (x)\geq 1$ when $x\geq x_1$.

Now pick $d_1$ and $\lambda_1$ such that $\lambda_1 d_1=x_1$. Construct an increasing, convex cost function $c_1\in C^1$ such that
$$c_1(\lambda)>E(\lambda;d_1)\quad \text{if } \lambda\neq \lambda_1$$
$$c_1(\lambda_1)=E(\lambda_1;d_1)$$
$$\frac{\partial E(\lambda; d_1)}{\partial \lambda}\bigg|_{\lambda=\lambda_1}=\frac{dc_1(\lambda)}{d\lambda}\bigg|_{\lambda=\lambda_1}.$$
As 
$$\frac{\partial^2 E(\lambda; d)}{\partial \lambda\partial d}\leq 0\quad \Leftrightarrow \quad \eta (\lambda d)\geq 1$$and $\eta (x)\geq 1$ for all $x\geq x_1$,
the cutoff transfer $d_1$ is the best among all cutoff transfers for cost function $c_1$. Since there exists an optimal transfer that is a cutoff transfer, cutoff transfer $d_1$ is the optimal transfer, and $\lambda_1$ is the maximum precision. By Lemma \ref{lem: monotone},
$$\eta (x)\leq 1\quad \text{if } x\leq x_1$$
$$\eta (x)\geq 1\quad \text{if } x\geq x_1.$$
Thus, $\eta (\cdot)$ single-crosses $1$ at $x_1$.

Similarly, for all $x_2>x_1$, I can pick $d_2$ and $\lambda_2$ such that $\lambda_2 d_2=x_2$. Pick a increasing, convex cost function $c_2\in C^1$ such that
$$c_2(\lambda)>E(\lambda;d_2)\quad \text{if } \lambda\neq \lambda_2$$
$$c_2(\lambda_2)=E(\lambda_2;d_2)$$
$$\frac{\partial E(\lambda; d_2)}{\partial \lambda}\bigg|_{\lambda=\lambda_2}=\frac{dc_2(\lambda)}{d\lambda}\bigg|_{\lambda=\lambda_2}.$$
By the same argument, we have
$$\eta (x)\leq \eta (x_2)\quad \text{if } x\leq x_2$$
$$\eta (x)\geq \eta (x_2)\quad \text{if } x\geq x_2.$$
Thus, $\eta (\cdot)$ is increasing once it goes above 1 at $x_1$.

\end{proof}

\begin{proof}[Proof of Proposition \ref{prop: cost}]

First suppose for all $\lambda$, $c_2(\lambda)\geq E(\lambda;d^*(c_1))$. Then $\bar d(c_2)\geq d^*(c_1)$, implying $d^*(c_2)\geq d^*(c_1)$. Moreover, since 
$$\text{for all }\lambda\geq \lambda^*(c_1) \text{ and } d\geq d^*(c_1),\quad\frac{\partial^2 E(\lambda; d)}{\partial \lambda\partial d}\leq 0$$ and $c_2-c_1$ is increasing,
$$\max [\arg\max_\lambda E(\lambda;d^*(c_2))-c_2(\lambda)]\leq\max [\arg\max_\lambda E(\lambda;d^*(c_1))-c_1(\lambda)].$$
Thus, $\lambda^*(c_2)\leq \lambda^*(c_1)$. Now suppose for some $\lambda$, $c_2(\lambda)\leq E(\lambda;d^*(c_1))$. Note that 
\[
\begin{split}
    &\max [\arg\max_\lambda E(\lambda;d^*(c_1))-c_2(\lambda)]\\
    =&\max [\arg\max_\lambda E(\lambda;d^*(c_1))-c_1(\lambda)-(c_2(\lambda)-c_1(\lambda))]\\
    \leq &\max [\arg\max_\lambda E(\lambda;d^*(c_1))-c_1(\lambda)]=\lambda^*(c_1),
\end{split}
\]where the inequality follows by $c_2-c_1$ being increasing. Thus, $$\bm\lambda(d^*(c_1);c_2)d^*(c_1)\leq \lambda^*(c_1) d^*(c_1)= \eta^{-1}(1).$$
By the proof of Theorem \ref{thm: optimal step function}, $d^*(c_2)\geq d^*(c_1)$ and $\lambda^*(c_2)\leq \lambda^*(c_1)$.

\end{proof}

\begin{proof}[Proof of Corollary \ref{col: comparative signal}]
    Changing the noise from $\varepsilon_1$ to $k\varepsilon_1$ is equivalent to keeping the noise at $\varepsilon_1$ and changing the cost function from $c_1(\lambda)$ to $c_2(\lambda)=c_1(k\lambda)$. Notice that $$c_2(\lambda)-c_1(\lambda)=c_1(k\lambda)-c_1(\lambda)$$
    which is increasing in $\lambda$ due to the convexity of $c_1$. The conclusion follows by Proposition \ref{prop: cost}.
\end{proof}

\begin{proof}[Proof of Proposition \ref{prop: high dimensional}]
The proof of the first part remains the same as Theorem \ref{thm: step}. For proving the second part, the only difference is to note that the expected transfer is
$$E(\lambda; d)=\int_0^{\lambda d} \phi(r) \frac{\pi^{n/2}}{\Gamma (n/2+1)}dr^n,$$
where $\Gamma$ is the gamma function and $\frac{\pi^{n/2}}{\Gamma (n/2+1)}r^n$ is the volume of the n-dimensional ball.
$$\frac{\partial^2 E(\lambda; d)}{\partial \lambda\partial d}=\frac{n\pi^{n/2}}{\Gamma (n/2+1)}\phi(\lambda d)(\lambda d)^{n-1}[n-\eta (\lambda d)].$$
Thus, 
$$\frac{\partial^2 E(\lambda; d)}{\partial \lambda\partial d}\geq 0\quad \Leftrightarrow \quad \eta (\lambda d)\leq n .$$
The rest of the argument remains the same.

\end{proof}

\begin{proof}[Proof of Proposition \ref{prop: Gaussian}]
The optimality of cutoff transfers follows by the Proof of Theorem \ref{thm: step} (3) implies (1), as Gaussian distributions satisfy increasing elasticity.

For proving the counterpart of Theorem \ref{thm: optimal step function}, it suffices to notice that
$$\frac{\partial^2 E(\Lambda(\lambda); d)}{\partial \lambda\partial d}=2\phi(\Lambda (\lambda)d)\frac{d\Lambda(\lambda)}{d\lambda}[1-\eta (\Lambda (\lambda)d)].$$
\end{proof}

\begin{proof}[Proof of Proposition \ref{prop: unobserved}]
    I first prove the case of uniform prior. Given the agent's signal $s$, his posterior belief about the state is $\mathcal N(s,1/\lambda^2)$. Since the principal's signal is Gaussian centered at $\theta$ with variance $\frac{1}{\lambda_p^2}$, the agent's posterior belief about $s_p$ is
    $$\mathcal N(s,\frac{1}{\lambda^2}+\frac{1}{\lambda_p^2}).$$
    Now the agent's payoff maximization problem is same as in Section \ref{sec: main result} except that $s_p$ replaces $\theta$ and his precision is $(\frac{1}{\lambda^2}+\frac{1}{\lambda_p^2})^{-\frac{1}{2}}$. Thus, Proposition \ref{prop: unobserved} follows from Theorem \ref{thm: step} and \ref{thm: optimal step function}.

    The proof for the Gaussian prior is similar. The only difference is that conditional on a signal $s$, the agent's posterior about $s_p$ is
    $$\mathcal N(\frac{s\lambda^2}{\lambda_0^2+\lambda^2}, \frac{1}{\lambda_0^2+\lambda^2}+\frac{1}{\lambda_p^2}).$$
    Thus, the precision is $\Lambda=[\frac{1}{\lambda_0^2+\lambda^2}+\frac{1}{\lambda_p^2}]^{-\frac{1}{2}}$.
\end{proof}

\section{Omitted Proofs of Section \ref{sec: savings}}
\label{subsec: savings proofs}

I first establish the regularity needed to compare arbitrary transfer
rules. Recall that $\hat E(\lambda;t)$ denotes the expected transfer
when the agent reports truthfully. The maximum over reports defining
$E(\lambda;t)$ is attained by Lemma~\ref{lem: expected transfer}.

\begin{lem}\label{lem: savings regularity}
The map $\lambda\mapsto\varphi(\cdot;0,\lambda)$ is continuously
differentiable as an $L^1(\mathbb R)$-valued map on $(0,+\infty)$,
with
\[
\frac{\partial\varphi(x;0,\lambda)}{\partial\lambda}
=\phi(\lambda x)+\lambda x\phi'(\lambda x).
\]
Consequently, for every bounded measurable transfer $t$,
\[
\frac{\partial\hat E(\lambda;t)}{\partial\lambda}
=\int_{\mathbb R}t(x)
  \frac{\partial\varphi(x;0,\lambda)}{\partial\lambda}\,dx.
\]
For every admissible transfer rule, $E(\lambda;t)$ tends to zero as
$\lambda$ tends to zero.
\end{lem}

\begin{proof}
Single-peakedness and integrability imply that $x\phi(x)$ tends to
zero as $x$ tends to $+\infty$, since
\[
x\phi(x)\leq 2\int_{x/2}^x\phi(y)\,dy\longrightarrow0.
\]
Integration by parts, symmetry, and $\phi'(x)\leq0$ for $x>0$ give
\[
\int_{\mathbb R}|x\phi'(x)|\,dx
=2\int_0^{+\infty}-x\phi'(x)\,dx=1.
\]
Thus $g(x)=\phi(x)+x\phi'(x)$ belongs to $L^1(\mathbb R)$.
The dilation map $\lambda\mapsto g(\lambda\cdot)$ is continuous
in $L^1$ for $\lambda>0$. This follows first for continuous functions
of compact support and then by approximation in $L^1$.

Pointwise differentiation gives
$\partial_\lambda\varphi(x;0,\lambda)=g(\lambda x)$. For sufficiently
small nonzero $h$, the fundamental theorem of calculus implies
\[
\left\|
\frac{\varphi(\cdot;0,\lambda+h)-\varphi(\cdot;0,\lambda)}{h}
-g(\lambda\cdot)
\right\|_1
\leq\frac{1}{|h|}
\int_{\min\{\lambda,\lambda+h\}}^{\max\{\lambda,\lambda+h\}}
\|g(u\cdot)-g(\lambda\cdot)\|_1\,du.
\]
The right-hand side tends to zero. This proves $L^1$
differentiability and continuity of the derivative. Integrating
against a bounded $t$ proves the formula for $\hat E$.

Finally, fix $\epsilon>0$ and choose $A>0$ such that
$t(x)\leq\epsilon$ whenever $|x|>A$. Since
$\varphi(x;0,\lambda)\leq\lambda\phi(0)$, for every report
displacement $b$,
\[
\int_{\mathbb R}t(x-b)\varphi(x;0,\lambda)\,dx
\leq2A\lambda\phi(0)+\epsilon.
\]
The bound is uniform in $b$, so $E(\lambda;t)\to0$ as
$\lambda\downarrow0$.
\end{proof}

\begin{proof}[Proof of Theorem~\ref{thm: cost minimizing cutoff},
$(3)\Rightarrow(1)$, and Theorem~\ref{thm: minimum expected transfer}]
Suppose that $\phi$ satisfies increasing elasticity below $1$.
Fix a cost function satisfying Assumption~\ref{ass: savings cost}. If $\bar\lambda c'(\bar\lambda)=0$, convexity and $c(0)=0$ imply
$c(\bar\lambda)=0$. The zero transfer therefore implements
$\bar\lambda$ and satisfies participation. Since transfers are
nonnegative, it is optimal and $C(\bar\lambda)=0$.
Moreover, $H(0)=G(0)=0$, so the formulas in
Theorem~\ref{thm: minimum expected transfer} hold with
$z_{\bar\lambda}=d^*(\bar\lambda)=0$.
Henceforth, suppose that $\bar\lambda>0$ and
$c'(\bar\lambda)>0$.

\paragraph{Step 1: a bound on marginal incentives.}
Extend $\eta$ to zero by $\eta(0)=0$, and write
\[
z_0=\inf\{z>0:\eta(z)\geq1\}.
\]
The function $H$ vanishes at zero and tends to zero at the positive
support endpoint, or at infinity when the support is unbounded.
It is strictly positive in the interior of the positive half-support.
It therefore attains a positive maximum at an interior point.
Moreover,
\[
H'(z)=2\phi(z)[1-\eta(z)]
\]
where $\phi(z)>0$. Increasing elasticity below $1$ implies that
$z_0$ is finite, lies in the positive support, and satisfies
\[
\eta(z_0)=1,\qquad
\eta(z)<1\quad(0\leq z<z_0),\qquad
\eta(z)\geq1\quad(z>z_0).
\]
Hence $H$ is strictly increasing on $[0,z_0]$, weakly decreasing
thereafter, and
\[
H(z_0)=\max_{z\geq0}H(z).
\]
The first point where elasticity reaches $1$ is used here, so this
argument also covers an interval on which $\eta=1$.

On the positive support of $\varphi(\cdot;0,\bar\lambda)$, define
\[
q_{\bar\lambda}(x)
=\frac{1-\eta(\bar\lambda|x|)}{\bar\lambda},
\]
and set $q_{\bar\lambda}(x)=0$ outside that support. Then
\[
\left.
\frac{\partial\varphi(x;0,\lambda)}{\partial\lambda}
\right|_{\lambda=\bar\lambda}
=q_{\bar\lambda}(x)\varphi(x;0,\bar\lambda)
\]
almost everywhere. The convention outside the support avoids
multiplying an infinite elasticity by a zero density.

Consider any transfer $t$ feasible in the relaxed problem
\eqref{eq: cost minimization}. If $t$ depends separately on the
state and the report, use the auxiliary transfer $\tilde t$
constructed in Lemma~\ref{lem: without loss}, with
$\bar\lambda$ in place of $\bm\lambda(t)$ as explained after
that lemma. If $t$ depends only on the prediction error, choose
an optimal report displacement $b^*$ at the target and let
$\tilde t(x)=t(x-b^*)$. In either case,
\[
\hat E(\bar\lambda;\tilde t)=E(\bar\lambda;t),\qquad
\hat E(\mu;\tilde t)\leq E(\mu;t)\quad\text{for all }\mu\geq0.
\]
The IC constraint therefore implies that $\bar\lambda$ maximizes
$\hat E(\mu;\tilde t)-c(\mu)$. Lemma~\ref{lem: savings regularity}
and the interiority of $\bar\lambda$ give the necessary condition
\begin{equation}\label{eq: savings local IC proof}
\int_{\mathbb R}\tilde t(x)q_{\bar\lambda}(x)
\varphi(x;0,\bar\lambda)\,dx=c'(\bar\lambda).
\end{equation}
Thus the bound below applies without assuming that $t$ induces
truthful reporting. It also does not require differentiating the
report-optimized value $E$.

Because $0\leq\tilde t\leq1$, condition
\eqref{eq: savings local IC proof} implies
\begin{align}
c'(\bar\lambda)
&\leq\int_{\mathbb R}(q_{\bar\lambda}(x))_+
\varphi(x;0,\bar\lambda)\,dx\notag\\
&=\frac{2}{\bar\lambda}
\int_0^{z_0}\phi(z)[1-\eta(z)]\,dz
=\frac{H(z_0)}{\bar\lambda}.
\label{eq: savings capacity proof}
\end{align}
The last equality follows from
$(z\phi(z))'=\phi(z)[1-\eta(z)]$.
This proves the necessity of \eqref{eq: savings feasibility}.

\paragraph{Step 2: construction and global incentive compatibility.}
Suppose that \eqref{eq: savings feasibility} holds. Strict increase
of $H$ on $[0,z_0]$ gives a unique solution
$z_{\bar\lambda}\in(0,z_0]$ to
\eqref{eq: savings root}; this is also the smallest nonnegative
solution. Set $d^*=z_{\bar\lambda}/\bar\lambda$ and let $t^*$ be
the cutoff transfer with cutoff $d^*$.

By symmetry and single-peakedness, a centered interval has weakly
greater probability than any other interval of the same length.
Thus the cutoff induces truthful reporting at every precision, and
\[
E(\mu;t^*)=G(\mu d^*)\quad\text{for all }\mu\geq0.
\]
Moreover,
\[
\frac{d}{d\mu}G(\mu d^*)=2d^*\phi(\mu d^*),\qquad
\frac{d^2}{d\mu^2}G(\mu d^*)
=2(d^*)^2\phi'(\mu d^*)\leq0.
\]
Consequently, $G(\mu d^*)-c(\mu)$ is concave. The equation
defining $z_{\bar\lambda}$ gives
\[
2d^*\phi(\bar\lambda d^*)=c'(\bar\lambda),
\]
so $\bar\lambda$ is a global maximizer. Since truthful reporting
is optimal at every precision, this verifies IC including changes
in reporting after a precision deviation. Comparing with $\mu=0$
also verifies IR, because $G(0)-c(0)=0$.

\paragraph{Step 3: minimum payment below the boundary.}
First suppose that $\bar\lambda c'(\bar\lambda)<H(z_0)$.
Then $z_{\bar\lambda}<z_0$, and
\[
\rho=\frac{1-\eta(z_{\bar\lambda})}{\bar\lambda}>0.
\]
Increasing elasticity below $1$ implies, on the positive density
support,
\[
q_{\bar\lambda}(x)\geq\rho\quad\text{if }|x|\leq d^*,
\qquad
q_{\bar\lambda}(x)\leq\rho\quad\text{if }|x|>d^*.
\]
For any feasible $t$, take the auxiliary transfer $\tilde t$ constructed in Step 1.
The payment bounds imply the pointwise inequality
\[
[q_{\bar\lambda}(x)-\rho][\tilde t(x)-t^*(x)]\leq0
\]
almost everywhere under $\varphi(\cdot;0,\bar\lambda)$.
Both $\tilde t$ and $t^*$ satisfy the same marginal condition
\eqref{eq: savings local IC proof}. Integrating therefore gives
\[
-\rho\int_{\mathbb R}[\tilde t(x)-t^*(x)]
\varphi(x;0,\bar\lambda)\,dx\leq0.
\]
The integral equals $E(\bar\lambda;t)-E(\bar\lambda;t^*)$.
Since $\rho>0$,
\[
E(\bar\lambda;t)\geq E(\bar\lambda;t^*)=G(z_{\bar\lambda}).
\]
Thus $t^*$ solves the relaxed problem and induces truthful reporting.

\paragraph{Step 4: minimum payment at the boundary.}
Now suppose that $\bar\lambda c'(\bar\lambda)=H(z_0)$, so that
$z_{\bar\lambda}=z_0$. Equality must hold in
\eqref{eq: savings capacity proof}. For the auxiliary transfer
$\tilde t$ constructed in Step 1 from any feasible contract, this requires
\[
\int_{\mathbb R}(1-\tilde t)(q_{\bar\lambda})_+
\varphi(x;0,\bar\lambda)\,dx
+\int_{\mathbb R}\tilde t(q_{\bar\lambda})_-
\varphi(x;0,\bar\lambda)\,dx=0,
\]
where $q_- = \max\{-q,0\}$. Both terms are nonnegative.
Consequently, $\tilde t=1$ almost everywhere on positive-score
events and $\tilde t=0$ almost everywhere on negative-score events.
The positive-score region is $|x|<z_0/\bar\lambda$, so the
expected payment is at least $G(z_0)$. The proposed cutoff attains
this bound by also paying zero on zero-score events, apart from
irrelevant boundary points. This establishes optimality even if
$\eta=1$ on an interval.

Together, these steps prove $(3)\Rightarrow(1)$ in
Theorem~\ref{thm: cost minimizing cutoff}, and all the claims in
Theorem~\ref{thm: minimum expected transfer}.
\end{proof}

\begin{proof}[Proof of Theorem~\ref{thm: cost minimizing cutoff},
$(1)\Rightarrow(2)\Rightarrow(3)$]
The implication $(1)\Rightarrow(2)$ is immediate. I prove
$(2)\Rightarrow(3)$ by contraposition.

\paragraph{Step 1: a cutoff on the initial increasing branch.}
Suppose that increasing elasticity below $1$ fails. There exist
$0<x_1<y$ in the positive support of $\phi$ such that
\[
\eta(y)<1,\qquad \eta(y)<\min\{\eta(x_1),1\}.
\]
Choose $v$ such that
$\eta(y)<v<\min\{\eta(x_1),1\}$, and let $r$ be the first point
where $\eta(r)=v$. Continuity of elasticity on the positive
support and $\eta(0)=0$ imply
\[
0<r<x_1<y,\qquad 0<v<1,\qquad
\eta(x)<v\quad(0\leq x<r).
\]
In particular,
\begin{equation}\label{eq: savings strict derivative proof}
\phi'(r)=-\frac{v}{r}\phi(r)<0,
\end{equation}
and $H$ is strictly increasing on $[0,r]$.

Take the target $\bar\lambda=1$ and the linear cost
\[
c(\mu)=k\mu,\qquad k=H(r)=2r\phi(r)>0.
\]
Let $t_0$ be the cutoff transfer with cutoff $r$.

\paragraph{Step 2: the baseline has a unique joint optimum.}
For a precision $\mu>0$ and report displacement $b=a-s$, let
\[
V_0(\mu,b)
=\int_{-r}^{r}\mu\phi(\mu(x+b))\,dx-k\mu.
\]
Set $V_0(0,b)=0$. Symmetry and single-peakedness imply
$V_0(\mu,b)\leq V_0(\mu,0)$, while
\[
\frac{\partial V_0(\mu,0)}{\partial\mu}
=2r\phi(\mu r)-k.
\]
Weak monotonicity of $\phi$, together with its strict local
decrease at $r$, makes this derivative strictly positive for
$\mu<1$ and strictly negative for $\mu>1$. Thus precision one
is the unique optimum. At $\mu=1$, the centered interval
probability is nonincreasing in $|b|$, and its second derivative
at $b=0$ is $2\phi'(r)<0$. It is therefore strictly smaller
for every $b\neq0$. Hence $(1,0)$ is the unique global maximizer,
with strictly positive value
\begin{equation}\label{eq: savings baseline rent proof}
V_0(1,0)=G(r)-H(r)
=2\int_0^r[\phi(x)-\phi(r)]\,dx>0.
\end{equation}

\paragraph{Step 3: reduce expected payment while preserving the marginal incentive.}
Choose symmetric sets $A$ and $B$, each a finite union of bounded
intervals, with $A$ just inside $\pm r$ and $B$ around $\pm y$.
By continuity, these sets can be chosen disjoint and inside the
positive density support so that every value of $1-\eta(|x|)$
on $B$ is strictly greater than every such value on $A$, and all
these values are positive. Write
\[
\begin{aligned}
p_A&=\int_A\phi(x)\,dx,
&m_A&=\int_A[1-\eta(|x|)]\phi(x)\,dx,\\
p_B&=\int_B\phi(x)\,dx,
&m_B&=\int_B[1-\eta(|x|)]\phi(x)\,dx.
\end{aligned}
\]
These four quantities are positive, and $m_A/p_A<m_B/p_B$.
Let $\kappa=m_A/m_B$ and define
\begin{equation}\label{eq: savings perturbation proof}
t_\epsilon=t_0-\epsilon\mathbf{1}_A
+\epsilon\kappa\mathbf{1}_B.
\end{equation}
For $0<\epsilon<\min\{1,1/\kappa\}$, the transfer takes values
in $[0,1]$. It is symmetric and compactly supported.
Lemma~\ref{lem: savings regularity} gives
\begin{equation}\label{eq: savings preserved gradient proof}
\left.
\frac{\partial\hat E(\mu;t_\epsilon)}{\partial\mu}
\right|_{\mu=1}
=k-\epsilon m_A+\epsilon\kappa m_B=k.
\end{equation}
At the same time, the truthful expected payment strictly decreases:
\begin{align}
\hat E(1;t_\epsilon)
&=G(r)-\epsilon p_A+\epsilon\kappa p_B\notag\\
&=G(r)-\epsilon p_A
\left(1-\frac{m_A/p_A}{m_B/p_B}\right)<G(r).
\label{eq: savings strict saving proof}
\end{align}

\paragraph{Step 4: preserve the exact local joint optimum.}
Define the joint net payoff
\[
V_\epsilon(\mu,b)
=\int_{\mathbb R}t_\epsilon(x)\mu\phi(\mu(x+b))\,dx-k\mu
\]
for $\mu>0$, and set $V_\epsilon(0,b)=0$.
These functions are $C^2$ near $(1,0)$. Indeed, for any bounded
interval $[\ell,u]$,
\[
\int_\ell^u\mu\phi(\mu(x+b))\,dx
=\Phi(\mu(u+b))-\Phi(\mu(\ell+b)),
\]
which is $C^2$ because $\phi$ is $C^1$. Every transfer in
\eqref{eq: savings perturbation proof} is a finite linear
combination of interval indicators. At the baseline optimum,
\[
\nabla^2V_0(1,0)
=\begin{pmatrix}
2r^2\phi'(r)&0\\
0&2\phi'(r)
\end{pmatrix}
\]
is negative definite. Thus there is a convex neighborhood $N$ of
$(1,0)$, with compact closure in $(0,+\infty)\times\mathbb R$,
on which the baseline Hessian is uniformly negative definite.
Since $V_\epsilon=V_0+\epsilon W$ for a fixed $C^2$ function $W$,
the same is true for every sufficiently small $\epsilon$.
Equation~\eqref{eq: savings preserved gradient proof} keeps the
precision derivative zero at $(1,0)$, and symmetry keeps the
reporting derivative zero there. Therefore $(1,0)$ remains the
unique maximizer of $V_\epsilon$ within $N$.

\paragraph{Step 5: rule out distant joint deviations.}
All sufficiently small perturbations have support in one common
interval $[-D,D]$ and take values in $[0,1]$.
For all report displacements $b$,
\[
V_\epsilon(\mu,b)\leq2D\phi(0)\mu
\quad\text{as }\mu\downarrow0,
\qquad
V_\epsilon(\mu,b)\leq1-k\mu
\quad\text{as }\mu\to+\infty.
\]
Moreover, for $\mu\in[\delta,L]$, where $0<\delta<L<+\infty$,
and for $|b|>D$, symmetry and monotonicity of $\phi$ give
\[
V_\epsilon(\mu,b)
\leq2DL\phi(\delta(|b|-D))\longrightarrow0
\quad\text{as }|b|\to+\infty.
\]
These bounds are uniform over sufficiently small $\epsilon$.
The positive baseline payoff in
\eqref{eq: savings baseline rent proof} therefore implies that
all choices capable of matching the payoff at $(1,0)$ lie in
one common compact rectangle $K$. Choose $K$ to contain
$\overline N$ in its interior. Uniqueness of the baseline maximum
and continuity give a strictly positive gap
\[
\Delta=V_0(1,0)-\max_{(\mu,b)\in K\setminus N}V_0(\mu,b)>0.
\]
Uniformly over the entire choice set,
\[
|V_\epsilon(\mu,b)-V_0(\mu,b)|
\leq\|t_\epsilon-t_0\|_\infty
\leq\epsilon\max\{1,\kappa\}.
\]
Choose $\epsilon$ small enough to preserve this gap and keep
$V_\epsilon(1,0)>0$. Step 4 rules out deviations inside $N$;
the gap rules out those in $K\setminus N$; and the uniform bounds
rule out those outside $K$. Thus $(1,0)$ is the unique global
maximizer of $V_\epsilon$.

The perturbed transfer is therefore feasible, induces truthful
reporting at the target, and satisfies strict IR. This argument
allows the agent to change his report after any precision deviation.
In particular, $E(1;t_\epsilon)=\hat E(1;t_\epsilon)<G(r)$.

\paragraph{Step 6: every feasible cutoff is more expensive.}
Any cutoff transfer with cutoff $d$ that implements precision one
under cost $k\mu$ must satisfy
\[
2d\phi(d)=k=H(r).
\]
Since $H$ is strictly increasing on $[0,r]$, no such cutoff has
$d<r$. Its expected payment is therefore at least $G(r)$.
By \eqref{eq: savings strict saving proof}, the globally feasible
transfer $t_\epsilon$ costs strictly less. Hence no cutoff transfer
is optimal for this linear cost and this implementable target.
This contradicts statement (2), proving $(2)\Rightarrow(3)$.
\end{proof}

\end{document}